\documentclass{article}
\usepackage{silence}
\PassOptionsToPackage{numbers, compress}{natbib}

\usepackage[preprint]{neurips_2026}

\usepackage[utf8]{inputenc} 
\usepackage[T1]{fontenc}    
\usepackage{hyperref}       
\usepackage{url}            
\usepackage{booktabs}       
\usepackage{amsfonts}       
\usepackage{nicefrac}       
\usepackage{microtype}      
\usepackage{xcolor}         

\newcommand{\mn}[1]{\textcolor{magenta}} 
\usepackage{graphicx}
\usepackage{tabularx}
\usepackage{amsthm}
\newtheorem{theorem}{Theorem}

\newtheorem{corollary}[theorem]{Corollary}
\newtheorem{proposition}[theorem]{Proposition}

\usepackage{amsmath}
\usepackage{multirow}
\usepackage{wrapfig}
\usepackage{times}
\usepackage{latexsym}
\usepackage{amsfonts}
\usepackage{booktabs}
\usepackage{tabularx}
\usepackage{graphicx}
\usepackage{subcaption}
\usepackage[table]{xcolor}
\usepackage{booktabs}
\usepackage{longtable}
\usepackage{array}

\newcommand{\bz}{\mathbf{z}}
\newcommand{\bZ}{\mathbf{Z}}

\title{The Confounder Trap: Treatment-Encoding Representations in Causal Inference with Text}

\author{%
  Marie Neubrander \\
  Department of Statistical Science\\
  Duke University\\
  Durham, NC, 27705 \\
   \And
  Graham Tierney \\
  Department of Statistical Science\\
  Duke University\\
  Durham, NC, 27705 \\
   \AND
   Alexander Volfovsky \\
   Department of Statistical Science\\
    Duke University\\
    Durham, NC, 27705 \\
}

\begin{document}
\maketitle
\begin{abstract}

Estimating causal effects of linguistic properties from observational text is difficult because the same document can contain both the treatment of interest and the non-treatment textual attributes needed for adjustment. Existing approaches often learn representations from the full text to capture latent confounding, but when treatment status is itself encoded by words in the text, these representations can directly encode treatment. This creates a confounder trap: richer representations can make treated and control documents separable, inducing overlap violations even when the underlying causal problem satisfies overlap. We study latent text treatments that are encoded through lexicons or other treatment-defining lexical information, and propose masking-based adjustment representations that remove this lexical treatment signal before representation learning. We formalize representation-induced overlap failure, prove that deletion masking preserves overlap for bag-of-words/topic-model representations, and characterize replacement masking as a natural relaxation for large language models that hides treatment-defining tokens while preserving word order and context. Across simulations, masking improves overlap diagnostics, stabilizes treatment effect estimates, and reduces bias relative to adjustment methods that learn from the unmasked text.
\end{abstract}

\section{Introduction}
\label{sec:intro}

In many applications, researchers want to estimate how linguistic properties of text causally affect reader behavior, attitudes, or downstream decisions. For example, prior work studies how positive product reviews affect sales \citep{pryzant-etal-2021-causal}, whether polite complaints receive faster responses \citep{pryzant-etal-2021-causal, gui2022causal}, how candidate descriptions change voter opinions \citep{fong-grimmer-2016-discovery}, and how information in clinical notes predicts patient responses to treatment\citep{mozer2024leveragingtextdatacausal}. These questions are text-as-treatment problems: the intervention of interest is not simply exposure to a document, but exposure to an encoded linguistic property.

Text-as-treatment problems are especially difficult because texts vary along many dimensions at once. A review can differ in sentiment, product category, specificity, writing quality, and author style. A political message can differ in humility, topic, argument structure, and ideological content. To isolate the effect of one linguistic property, the analyst must compare texts that differ in the treatment while adjusting for other textual attributes that also affect the outcome or the assignment of treatment; we refer to these as adjustment-relevant non-treatment textual attributes.


This creates a tension between two key assumptions in observational causal inference: ignorability and overlap \citep{rubin1983propensity}. The \textit{overlap} assumption requires that for alladjustment-relevant non-treatment attributes, there is a nonzero probability of the corresponding document being treated or control. Intuitively, this ensures we can compare treated and untreated units that are similar in their covariates.  At the same time, we must also satisfy \textit{ignorability}. This requires measuring and adjusting for enough of the non-treatment textual attributes to remove confounding. Overlap and ignorability have competing interests: we must capture enough confounding information to satisfy ignorability, while not so much that we directly encode treatment status. Once treatment is directly recoverable from the representation, treated and control documents occupy separated regions of the representation space and overlap fails.

To illustrate, suppose treatment is lexicon-based positive/negative sentiment in product reviews, and the outcome is whether a reader purchases the product. Product category may be adjustment-relevant: movies and books may differ in baseline purchase rates, and also in the distribution of positive and negative reviews (Figure~\ref{fig:trap} Panel A). But sentiment words such as "loved" and "hated" are not adjustment variables; they define treatment. An embedding learned from the full review will naturally encode them, enabling the representation to classify treatment status so well that it violates overlap and becomes the wrong object for causal adjustment (Figure~\ref{fig:trap} Panel B).

\begin{figure}[t]
\centering
\includegraphics[width=.9\linewidth]{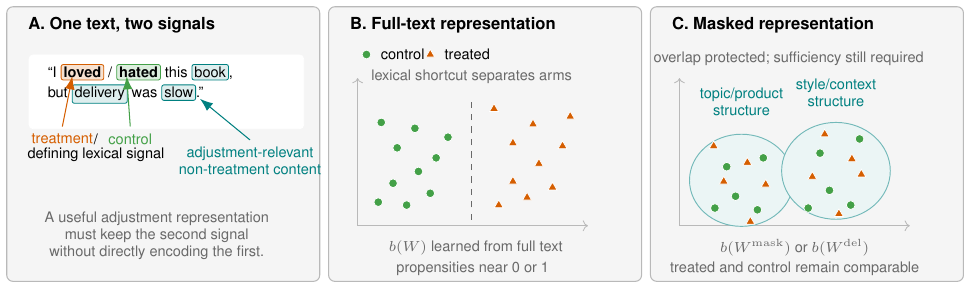}
\caption{The confounder trap. (A) A document contains both treatment-defining words and adjustment-relevant non-treatment content. (B) Representations learned from full text can collapse overlap. (C) Masking treatment-defining words before representation learning removes this separation while preserving non-treatment structure for adjustment.}
\label{fig:trap}
\end{figure}

A growing body of work learns text representations for causal adjustment using topic models, bag-of-words features, or large language models  \citep{gui2022causal, veitch2020adapting, pryzant-etal-2021-causal, li-etal-2025-using}. 
However, when the treatment is a component of the text, these representations can induce rather than solve a causal identification problem  \citep{katta2025llm}:  the learned representation encodes the treatment. We call this mechanism {direct treatment encoding}. 
We focus empirically on treatments with identifiable lexical treatment signals
such as sentiment lexicons, partisan-framing terms, or hedging words and propose a masking framework to avoid direct treatment encoding while learning representations sufficient for adjustment.

One natural response is to erase treatment information from learned representations post hoc. For instance, LEACE \citep{belrose2025leaceperfectlinearconcept} gives a closed-form linear projection that removes a target concept from a representation. Though designed for concept erasure rather than causal inference, it and similar representation-editing methods are natural candidates for this problem. However, erasing \emph{all} treatment signal is the wrong target: confounders are by definition treatment-predictive, so removing all treatment information also discards adjustment signal needed for identification. The challenge is to remove the \emph{lexical shortcut} while preserving confounding structure.

\textbf{Contributions. } First, we formalize representation-induced overlap failure: when a representation predicts treatment better than the true confounders allow, overlap vanishes. 
Second, for treatments with identifiable removable treatment signals,

we develop a masking framework with two variants (Figure~\ref{fig:trap} Panel C). Deletion masking provably preserves overlap; replacement masking relaxes this setting for use with LLMs. Third, we implement masking for both topic models and large language models, showing that it improves overlap, stabilizes treatment effect estimates, and reduces bias relative to unmasked methods.

Sections~\ref{sec:CB} and \ref{sec:theory} formalize the problem and framework. Sections~\ref{sec:topic} and \ref{sec:llm} give implementations and simulations for topic models and LLMs, respectively. Section~\ref{sec:real-data} gives a real-data application.

\section{Potential Outcomes and Text-as-Treatment}
\label{sec:CB}

We use the potential outcomes framework to define the causal effect of a latent linguistic treatment. Suppose an individual reads a document with words $\mathbf{W}_i$ from which they perceive a treatment of interest $T_i \in \{0,1\}$ alongside other textual attributes $\bZ_i$. After reading, an outcome $Y_i$ is realized.  We additionally assume an observed document label $D_i$ related to the treatment of interest, either assigned by an analyst or measured directly from the text.  Figure~\ref{fig:dag} shows the corresponding DAG: a label $D_i$ generates the words $\mathbf{W}_i$. The reader perceives latent treatment $T_i$ and confounding attributes $\bZ_i$ as deterministic functions of $\mathbf{W}_i$, both of which affect $Y_i$.
\begin{wrapfigure}{r}{0.38\columnwidth}
\centering
\includegraphics[width=0.37\columnwidth]{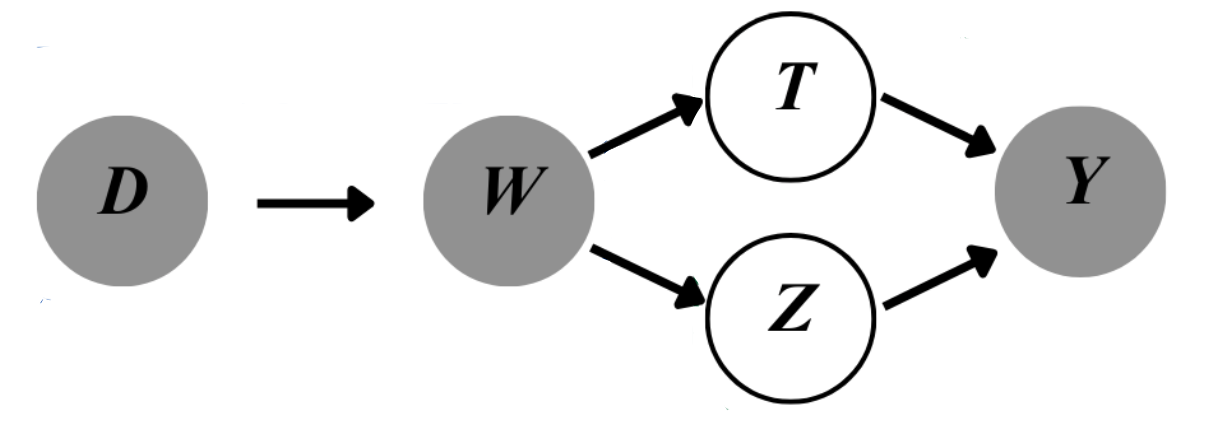}
 \caption{DAG. Grey nodes: observed outcome $Y$, labels $D$, and words ${W}$. White nodes: unobserved latent treatment $T$ and confounders ${Z}$. 
 }
  \label{fig:dag}
\end{wrapfigure}
We work under the simplifying alignment $D_i = T_i$: the label is assumed to match the latent perceived treatment.  
For treatments whose defining signal is observable in the text,
the assumption is natural: treatment status is determined by the presence of specific words, and it is reasonable to assume a reader correctly identifies those words.  This alignment allows us to identify the effect of changing the latent treatment while holding fixed the non-treatment textual attributes:
\begin{equation*}
\label{eq:tau_t}
    \tau_t = \mathbb{E}_\bZ\bigl[Y(T=1,\bZ) - Y(T=0,\bZ)\bigr].
\end{equation*}

This differs from the effect of assigning a label $D$, which may also change non-treatment attributes. Appendix \ref{app:textest} develops this distinction in more detail.

If $\bZ_i$ were observed, standard estimation methods, such as outcome regression or inverse propensity score weighting, could adjust for it.

The challenge is that $\bZ_i$ is not directly observed; analysts substitute a learned representation. We write $Y(t)$ for $Y(t, \bZ)$ and abbreviate $\tau_t$ as $\tau = \mathbb{E}[Y(1) - Y(0)]$.  Let $B = b(\mathbf{W})$ denote a learned representation of the document.  Identification of $\tau$ from the observed data requires four assumptions.

\textbf{A1: SUTVA} If $T = t$, then $Y = Y(t)$.  There is no interference  and no hidden versions of treatment.

\textbf{A2: Latent mean ignorability} For $t \in \{0,1\}$, $\mathbb{E}\bigl[Y(t) \mid T = t,\, \bZ,\, B\bigr] = \mathbb{E}\bigl[Y(t) \mid \bZ,\, B\bigr] .$ That is, conditional on the latent confounders $\bZ$ and the representation $B$, treatment assignment carries no additional information about mean potential outcomes.

\textbf{A3: Outcome sufficiency} For $t\in\{0,1\}$, $\mathbb{E}\bigl[Y(t) \mid Z,\, B\bigr] = \mathbb{E}\bigl[Y(t) \mid B\bigr].$ $B$ retains all relevant information from $Z$.

\textbf{A4: Representation level overlap} For $t \in \{0,1\}$,
$\Pr(T = t \mid B) > 0$. Any value of the representation $B$ has a nonzero probability of being associated with either treatment status.

\begin{theorem}
\label{thm:ident}
Under Assumptions~A1-A4, $\tau = \mathbb{E}_B\!\bigl[\mathbb{E}(Y \mid T=1,\, B) - \mathbb{E}(Y \mid T=0,\, B)\bigr].$
That is, $B$ is sufficient for mean adjustment for the ATE (proof in Appendix~\ref{app:th0}).
\end{theorem}
Our contribution concerns ensuring A4 is satisfied, which is uniquely fragile in text settings.

\section{Overlap and Masking}
\label{sec:theory}

Theorem~\ref{thm:ident} showed that any representation $B$ satisfying mean ignorability, sufficiency and overlap suffices for identification of $\tau$. The challenge in text settings is that the same words encode both treatment and confounders, so a representation flexible enough to capture confounding (satisfying the ignorability assumption) may also recover treatment exactly (violating the overlap assumption). 

The key object in our framework is a treatment-removing transformation of the text that preserves adjustment-relevant non-treatment content. We assume that such a transformation is available to the analyst and can be performed prior to the causal analysis pipeline. We assume that the transformed texts carry no direct treatment information beyond the latent adjustment attributes, meaning that $T$ is independent of the transformed text, conditional on $\mathbf{Z}$.
When this holds, any treatment predictability remaining in the transformed text is attributable to $\mathbf{Z}$, and is therefore legitimate adjustment signal rather than direct treatment encoding. Lexicon deletion is one concrete case in which such a transformation is easy to define.

\subsection{Overlap Bounds}
The \textit{true} propensity scores $e^{\star}(\mathbf{z}) = \Pr(T{=}1 \mid \mathbf{Z}{=}\mathbf{z})$ are unobservable
; 
we assume the true confounders satisfy overlap: $\alpha^{\star} \leq e^{\star}(\mathbf{Z}) \leq 1 - \alpha^{\star}$ for some $\alpha^{\star} \in (0, \frac{1}{2})$. In its place, a learned representation $b(\mathbf{W})$ induces the propensity score $e_b(\mathbf{w}) = \Pr(T{=}1 \mid b(\mathbf{w}))$. 
Bayes error is used to measure how much treatment information $b(\mathbf{W})$ carries:
\begin{equation}
\label{eq:bayes-error}
    R_b \;=\; \mathbb{E}_{\mathbf{W}}\!\big[\min\!\big(\Pr(T{=}1 \mid b(\mathbf{W})),\; 1 - \Pr(T{=}1 \mid b(\mathbf{W}))\big)\big].
\end{equation}
$R_b$ is the \textit{irreducible error} of predicting $T$ from $b(\mathbf{W})$: the lowest misclassification rate achievable by any classifier. It equals $0$ when $b(\mathbf{W})$ determines $T$ perfectly and $\min(\Pr(T{=}1), \Pr(T{=}0))$ when $b(\mathbf{W})$ is independent of $T$ (i.e., the best one can do is predict the majority class). It is model free, depending only on the joint distribution of $(T, b(\mathbf{W}))$.

Let $R^{\star} = \mathbb{E}_\mathbf{Z}[\min(e^\star(\mathbf{Z}), 1-e^\star(\mathbf{Z}))]$ denote the (unknown) Bayes error under the \textit{true} confounders $\mathbf{Z}$. The overlap assumption implies $R^{\star} > 0$. Comparing $R_b$ to $R^{\star}$ characterizes what a learned representation has captured. If  $R_b \approx R^{\star}$, then confounding is captured faithfully. If $R_b > R^{\star}$, then  some confounding is missed; ignorability may fail. If $R_b < R^{\star}$ then $b(\mathbf{W})$ classifies $T$ better than the true confounders allow, 
violating \textbf{A4}.
This is the scenario of interest.
\begin{theorem}[Overlap bound]
\label{thm:approx}
For any representation mapping $b$ and any $\alpha \in (0, 1/2)$,
$
\Pr\!\big(\alpha \leq e_b(\mathbf{W}) \leq 1 - \alpha\big) \;\leq\; R_b/\alpha.
$
Proof is in Appendix~\ref{app:th1}.
\end{theorem}

 For a fixed margin $\alpha$, the fraction of documents with propensity scores in the overlap region $[\alpha, 1{-}\alpha]$ is at most $R_b / \alpha$. The bound is informative when the representation is highly predictive of $T$, i.e., when $R_b \ll \alpha$. As $b(\mathbf{W})$ becomes increasingly able to separate treated from control documents, $R_b \to 0$ and overlap vanishes.
\begin{corollary}[Exact encoding]
\label{cor:exact}
If $f(b(\mathbf{W})) = T$ for some $f$, then $R_b = 0$ and $e_b(\mathbf{W}) \in \{0,1\}$ with probability 1.
\end{corollary}
We note that the statements above refer to the true propensity scores $e_b$, which are never observed. In practice, $e_b$ must be estimated from data. The Bayes error $R_b$ is likewise an unobservable quantity, but any classifier trained to predict $T$ from $b(\mathbf{W})$ provides an upper bound. If the held-out classification error is small, then $R_b$ is at most as small, and Theorem~\ref{thm:approx} guarantees that overlap is limited.

\subsection{Masking}
\label{sec:masking}
For lexicon-based treatments, the text transformation can be constructed by deleting or replacing treatment-defining words. For non-lexical treatments, such as politeness or confidence, the same logic applies if the analyst can identify textual cues that directly define the treatment and remove them from the representation-learning input. The theory below is therefore best read as a guarantee for any transformation satisfying that $T$ is independent of the transformed text, conditional on $\mathbf{Z}$, with lexicon masking providing the cleanest implementation.

Theorem \ref{thm:approx} implies a design goal: an adjustment representation should be constructed so that $R_b$ remains bounded away from zero. 

Lexicon-based treatments arise naturally in many text-as-treatment settings and are appealing for several reasons.

Treatment status is directly measurable from the text, the assumption $D = T$ holds by construction, 
and the treatment effect has a clear actionable interpretation: a writer can impact a desired outcome via  word choice.

Recall the example in Figure~\ref{fig:trap}. Consider the reviews ``I loved this book!" and ``I hated this movie!"  Embeddings of these documents will lie in different regions of the representation space, separated by their opposing sentiment treatment status. Now suppose we mask these lexicon words, producing "I [MASK] this book!" and "I [MASK] this movie!" Embeddings can still capture confounding content but won't separate units by sentiment words alone. In the language of Theorem \ref{thm:approx}, masking keeps propensity scores bounded away from 0 and 1 by preventing $R_b$ from shrinking toward zero.

\begin{theorem}[Preserving overlap]\label{prop:masking}Let $\mathbf{W}^{-}$ be a treatment-removing transformation of $\mathbf{W}$ satisfying
\[
T \perp \mathbf{W}^{-}\mid \mathbf{Z}.
\]
Let $b_{-}(\mathbf{W}) = h(\mathbf{W}^{-})$ be any representation that depends on $\mathbf{W}$ only through $\mathbf{W}^{-}$.
Then
\[
\alpha^\star \leq P(T=1\mid b_{-}(\mathbf{W})) \leq 1-\alpha^\star.
\]
\end{theorem}

In particular, every document satisfies overlap with margin $\alpha^{\star}$, so the failure mode identified by Theorem~\ref{thm:approx} is structurally avoided; the proof is in Appendix~\ref{app:th3}

We now define the simplest transformation: complete deletion of lexicon words.  Suppose treatment is determined by a known lexicon $\mathcal{L}$, so that there exists some $g$ such that $T = g(\mathbf{W} \cap \mathcal{L})$, and let $\mathbf{W}^{\mathrm{del}} = \mathbf{W} \setminus \mathcal{L}$ denote the document with lexicon words removed. 
Specializing the above theorem to $W^{del}$, we rely heavily on the assumption $T \perp \mathbf{W}^{\mathrm{del}} \mid \mathbf{Z}$. In unmasked text, this is not true: the presence of words in $\mathcal{L}$ encodes treatment status. After deletion, $\mathbf{W}^{\mathrm{del}}$ no longer directly encodes $T$, and the assumption holds whenever $\mathbf{Z}$ captures any remaining treatment-correlated variation in non-lexicon text. This is the relevant case for bag-of-words topic modelling in Section \ref{sec:topic}.

Deletion masking gives a clean theoretical guarantee but hides information, such as  word positions, context around treatment words, and document length, that may be relevant for adjustment. For large language model representations, we therefore adopt a replacement masking strategy. Define $\mathbf{W}^{\mathrm{mask}}$ as the document with lexicon words \textit{replaced} by \texttt{[MASK]} tokens rather than deleted. It is no longer clear that  $T \perp \mathbf{W}^{\mathrm{mask}} \mid \mathbf{Z}$ should hold.

Replacement masking is best understood as a training-time intervention. Let $\theta_M$ denote parameters learned from masked training inputs, and let $b_{\theta_M}(\cdot)$ denote the resulting representation map. We define the mask-trained representation evaluated on the original (unmasked) text,
$b_M(\mathbf{W}) := b_{\theta_M}(\mathbf{W}),$
with induced propensity score $e_M(\mathbf{W}) := \Pr(T = 1 \mid b_M(\mathbf{W}))$. We evaluate on the unmasked text as our aim is not to forget that lexicon words exist. Rather, we aim to prevent the encoder from learning a deterministic shortcut from those tokens to treatment status during training while preserving as many adjustment-relevant textual attributes as possible at evaluation.

\begin{proposition}[Full text evaluation can preserve overlap.]
\label{prop:eval-full}
There exist data-generating processes where $b_M(\mathbf{W}^{\mathrm{mask}})$ violates overlap while $b_M(\mathbf{W})$ satisfies it.
\end{proposition}

\noindent Intuitively, when the mask pattern itself reflects $T$, for instance, if treated documents contain one lexicon token and controls contain none, evaluating on $\mathbf{W}^{\mathrm{mask}}$ exposes that signal directly. A mask-trained encoder evaluated on $\mathbf{W}$, by contrast, can fall back on non-lexicon content and avoid the lexicon-to-treatment shortcut. Appendix~\ref{app:p1} gives an explicit DGP and proof.

Because the strict conditional independence assumption no longer holds, we cannot guarantee that $e_M(\mathbf{W})$ remains in $[\alpha^{\star}, 1 - \alpha^{\star}]$. Appendix~\ref{app:replacement-bound} formalizes a sensitivity bound showing that overlap degrades smoothly with residual leakage rather than collapsing as in the unmasked case.
The bound is not directly estimable, but it makes precise the way that replacement masking is a controlled relaxation of deletion masking. In practice, we assess residual leakage by reporting the fraction of estimated propensity scores in $[\alpha,\, 1 - \alpha]$  for a chosen margin $\alpha$.

We implement the masking framework in two settings of increasing representational complexity. Section~\ref{sec:topic} uses topic models, whose bag-of-words structure is a natural fit for deletion masking: removing lexicon words leaves no trace of their count or position, so Theorem~\ref{prop:masking} applies directly and gives a clean proof of concept. Section~\ref{sec:llm} moves to fine-tuned large language models, where confounding can operate through word order, syntax, and contextual semantics that topic proportions do not capture. Because deletion would destroy this structure, we adopt replacement masking.

\section{Topic Modeling}
\label{sec:topic}

Topic models provide low-dimensional, interpretable summaries of text and remain widely used in applied causal inference with text \citep{falav207exercise, roberts2020matching, sridhar2019debates, ahrens-etal-2021-bayesian}, and analysts who already use LDA can adopt the masking pipeline with minimal modification.

In the spirit of the causal amortized topic model \citep{veitch2020adapting}, we assume documents contain $K + 2$ latent topics. Topics $K+1$ and $K + 2$ map to treatment and control words, respectively; presence or absence of these words indicate treatment status. The remaining $K$ topics are potential confounders. Outcomes $Y$ are assumed to be some function of the $K + 2$ topics. {In the notation of Section~\ref{sec:masking}, the $K$ confounding topics correspond to confounders $\mathbf{Z}$, while topics $K+1$ and $K+2$ encode treatment status $T$. The document's words $\mathbf{W}$ are observed; $T$ and $\mathbf{Z}$ must be recovered.}

The topic model is as follows: for each document $i \in \{1, \ldots, N\}$,  draw treatment label $D_i \sim \text{Bern}(p)$ and document length $L_i \sim \text{Pois}(\mathcal{L})$. A topic proportion parameter is drawn dependent on treatment status $\boldsymbol{\theta}_i \sim \text{Dirichlet}(D_i\alpha_T + (1-D_i)\alpha_C)$ where $\alpha_T \neq \alpha_C$ are topic distribution parameters. For each word $l \in \{1, \ldots, L_i\}$, a topic $q_{il}$ is drawn $q_{il} \sim \text{Multinomial}(\boldsymbol{\theta}_i$). Word $w_{il}$ is then drawn from a corresponding topic distribution $\boldsymbol{\beta}_{q_{il}}$. A plate diagram is provided in \ref{app:LDA}.

\subsection{Estimation Procedure}
\label{sec:estimation}

Fitting LDA with $K+2$ topics to learn confounders appears natural, but two of the resulting proportions encode treatment status, making $T$ perfectly recoverable and collapsing overlap by Corollary~\ref{cor:exact}. Our pipeline identifies treatment-defining words via domain knowledge or by inspecting the top words of treatment/control topics from an initial $K+2$ fit, deletes them from each document's bag-of-words, and fits LDA with $K$ topics on the masked text. Because LDA is bag of words, $T \perp \mathbf{W}^{\mathrm{del}} \mid \mathbf{Z}$ holds naturally, and Theorem~\ref{prop:masking} provides that overlap is preserved. To avoid bias, models are fit on a training split and the ATE is estimated on a test set \citep{egami2018makecausalinferencesusing}.
\subsection{Simulations}

\label{sec:sims}
Since true causal effects are unobserved in real-world data, we use semi-synthetic simulations to establish a ground truth. 
In each iteration of a simulation, we generate $N$ documents according to the model above, conditioning on treatment or control words being present in each document. Using the drawn topic values $q_{il}$, we compute the true topic proportions for each document $i$ given by vector ${\bz_i} \in [0, 1]^{K+2}$. With true treatment effect $\tau$ and coefficients $\boldsymbol{\gamma} \in \mathbb{R}^{K+2}$, we generate a document's outcome $y_i \sim N(\tau D_i+ \boldsymbol{\gamma}^{T}{\bz_i}, \sigma^2)$. 

We compute several ATE estimates: (1) \textbf{naive}, the difference in means between treatment and control groups; (2) \textbf{or\_full}, \textbf{or\_leace}, and \textbf{or\_masked}, outcome regression estimates; (3) \textbf{ipw\_full},  \textbf{ipw\_leace}, and \textbf{ipw\_masked}, Hájek IPW estimates; and (4) \textbf{aipw\_full}, \textbf{aipw\_leace} and \textbf{aipw\_masked}, AIPW estimates. Appendix~\ref{app:estimators} gives the full form of each estimator. The \textbf{*\_full} estimators use confounders from fitting LDA with $K + 2$ topics, the \textbf{*\_masked} estimators use our pipeline, and the  \textbf{*\_leace} estimators apply LEACE to the full $K+2$ proportions to remove linear treatment signal before outcome and propensity modeling.
Propensity scores are estimated by logistic regression and winsorized to $[0.1, 0.9]$ following \citet{crump2009dealing}. We mask treatment/control words  identified from the top words of an initial $K+2$-topic LDA fit, then refit LDA with $K$ topics on the masked text using an 80/20 train/test split. Implementation details are in Appendix~\ref{app:all_lda}.

\textbf{Parameters. }
This section presents results from one illustrative set of parameters. Document generation ($\beta, \alpha_T, \alpha_C$) is developed to reflect realistic data based on a data set of Associated Press articles in R \texttt{topicmodels} package \citep{grun2011tm}. Full details, plus simulations using more parameters and topic misspecification stress-testing, are in Appendix \ref{app:all_lda}. 



\begin{table}[h]
\centering
\small
\caption{MSE, bias, and variance of ATE estimators (true $\tau = 0.5$) across 1,000 simulations.}
\label{tab:decomp}
\begin{tabular}{@{}l c ccc ccc ccc@{}}
\toprule
 & \textbf{naive}
 & \multicolumn{3}{c}{\textbf{IPW}}
 & \multicolumn{3}{c}{\textbf{OR}}
 & \multicolumn{3}{c}{\textbf{AIPW}} \\
\cmidrule(lr){3-5}\cmidrule(lr){6-8}\cmidrule(lr){9-11}
          &        & full   & masked          & leace  & full   & masked          & leace  & full   & masked          & leace  \\
\midrule
MSE       & .01462 & .01464 & \textbf{.00300} & .01462 & .00606 & {.00287} & .01462 & .00697 & \textbf{.00300} & .01462 \\
Bias      & .110   & .0895  & \textbf{.00853} & .110   & .0189  & \textbf{.00927} & .110   & .0198  & \textbf{.00891} & .110   \\
Var       & .00258 & .00664 & {.00293} & .\textbf{00258} & .00570 & {.00278} & .\textbf{00258} & .00658 & {.00293} & .\textbf{00258} \\
\bottomrule
\end{tabular}
\end{table}

\textbf{Results. }Results across 1,000 simulation iterations are summarized in Table~\ref{tab:decomp}. Masking reduces MSE, bias, and variance for every estimator. For OR and AIPW the improvement in MSE comes mostly from reduced variance, while for IPW both bias and variance fall sharply.  The full $K+2$ IPW  estimator performs nearly as poorly as the naive baseline. Iteration-level comparisons (Appendix \ref{app:extra_lda} Figure~\ref{fig:rmse-app}) confirm the gain is consistent across runs: the masked estimator achieves lower squared error than the full estimator in 64\%, 62\%, and 83\% of iterations for OR, AIPW, and IPW respectively.
The LEACE estimators recover the naive estimator across all three adjustment strategies. LEACE erases all linear treatment signal from the $K+2$ topic proportions, including necessary adjustment information.

An overlap diagnostic is the empirical counterpart of Theorem~\ref{thm:approx}. The full $K+2$ representation produces propensity scores outside $[0.1, 0.9]$ for nearly every document in nearly every iteration: the mean fraction of documents with extreme propensities across simulations is 99.98\%, and even in the most favorable iteration, 96.2\% of documents are extreme. Masking reduces this fraction to 1.7\% on average.  LEACE drives it to essentially zero at the cost of stripping the representation of adjustment. Table \ref{tab:overlap-app} in Appendix \ref{app:extra_lda} provides a summary of the distributions of extreme propensity scores.

\section{Large Language Models}
\label{sec:llm}

When confounding operates through context, syntactic structure, or semantic content not captured by topic proportions, a richer representation than topic mixtures is needed. Following prior work, we therefore implement the masking framework with a fine-tuned large language model \citep{veitch2020adapting, pryzant-etal-2021-causal, gui2023causalestimationtextdata}.

\subsection{Architecture and Masked Input}
We fine-tune DistilBERT: Each document $\mathbf{W}_i$ is tokenized into a sequence of subword tokens (using \texttt{DistilBertTokenizer}) and prefixed by the special \texttt{[CLS]} token \citep{sanh2019}. Following standard BERT practice, we use the final-layer hidden state at the \texttt{[CLS]} position as the document embedding.
Let $\mathbf{W}_i^{\mathrm{mask}}$ denote the document obtained by replacing every token belonging to a treatment- or control-defining lexicon entry with the \texttt{[MASK]} token to leverage the encoder's existing masked-language-model capabilities. Multi-token lexicon entries are masked together. 

During fine-tuning, the encoder is applied to the masked input, yielding the representation $b_\theta^{M_i} := b_\theta(\mathbf{W}_i^{\mathrm{mask}})$. This representation feeds three heads: an outcome head $\widetilde{Q}\!(t, b_\theta^{M_i})$, a propensity head $\widetilde{g}\!(b_\theta^{M_i})$, and a masked-language-modeling (MLM) head. The outcome and propensity heads are used for causal estimation, while the MLM head preserves general language understanding.

The masking intervention is applied only to the text-derived representation. When observed non-text covariates $C_i$ are available, we pass them directly to the outcome and propensity heads by replacing $b_\theta^{M_i}$ with $(b_\theta^{M_i}, C_i)$ \citep{pryzant-etal-2021-causal, gui2023causalestimationtextdata}. We suppress $C_i$ below as the overlap issue studied here arises from learned text representations, not ordinary observed covariates.

\textbf{Multi-task loss. }
The \textbf{outcome head }$\widetilde{Q}\!(t, b_\theta^{M_i})$ estimates $ \mathbb{E}\!\left[Y \mid T=t,\, b_\theta(\mathbf{W}^{\mathrm{mask}})\right]$
for $t \in \{0,1\}$, with loss $L_{\mathrm{out}}$ as mean squared error for continuous outcomes and cross-entropy for binary outcomes. The \textbf{propensity head }$\widetilde{g}\!(b_\theta^{M_i})$ estimates $\Pr\!\left(T=1 \mid b_\theta(\mathbf{W}^{\mathrm{mask}})\right)$
and contributes two loss terms playing complementary roles. The first, $L_{\mathrm{CE}}$, is the cross-entropy loss for predicting $T_i$ from $b_\theta^{M_i}$. This term is necessary because genuine confounders are generally treatment-predictive: if the representation were never encouraged to learn treatment signal, it may not learn adjustment-relevant information. The second term is a soft overlap penalty that activates only when the estimated propensity leaves a target overlap band $[\alpha,1-\alpha]$: $
L_{\mathrm{penalty}}
=
\mathbb{E}\!\left[
(\alpha-\widetilde{g})^2 \mathbb{I}\{\widetilde{g}<\alpha\}
+
(\widetilde{g}-(1-\alpha))^2 \mathbb{I}\{\widetilde{g}>1-\alpha\}
\right].
$

The \textbf{MLM head} contributes $L_U$ by randomly masking non-lexicon tokens and training the encoder to recover them. We exclude already masked lexicon tokens from this auxiliary masking procedure. The full training loss combines all four terms: $$ L(w_i; \theta) = \lambda_{Q}\, L_{\text{out}} + L_{\text{CE}} +  \lambda_{P}\, L_{\text{penalty}} + L_U.$$

The weight $\lambda_Q$ is auto-calibrated at initialization so that the outcome and cross-entropy losses contribute equally to the gradient. The penalty weight $\lambda_P$ is tuned adaptively during training, targeting a 90\% validation overlap rate.\footnote{90\% is chosen by analogy to St\"{u}rmer trimming \citep{sturmer2021propensity}.} Full update rules are in Appendix~\ref{app:lambda-details}.

\textbf{Training and Estimation. }
We use $80\%$/$10\%$/$10\%$ training/validation/test splits; fine-tuning uses only the training split. Optimization uses AdamW with learning rate $2 \times 10^{-5}$, batch size $32$, and up to $20$ epochs. At the end of each epoch we record  validation losses and select the epoch with the minimum value.

Given a fitted model at the selected checkpoint, the ATE is estimated by outcome regression on the held-out test split:
$\widehat{\tau}^{\,\text{mask}} \;=\; \frac{1}{n_{\text{test}}} \sum_{i \in \text{test}} \big[\widetilde{Q}(1,b_\theta(\mathbf{W}_i)) \;-\; \widetilde{Q}(0,b_\theta(\mathbf{W}_i))\big].$

\textbf{Comparators. }We compare against two DistilBERT-based estimators. \textbf{TextCause} \citep{pryzant-etal-2021-causal} trains outcome heads with an MLM loss  but does not use a treatment-prediction loss. We therefore use it as a baseline for outcome-focused representation learning. \textbf{TI-estimator} \citep{gui2023causalestimationtextdata} jointly trains outcome and propensity heads with an MLM loss and uses the learned outcome summaries in AIPW and outcome regression estimators. It therefore tests whether standard treatment-aware representation learning is sufficient when the treatment-defining lexicon remains visible to the encoder. While \citet{gui2023causalestimationtextdata} recommend their AIPW estimator, we also report outcome regression as it is most similar to our approach. Implementation details are in Appendix~\ref{app:llm-comp}.

\subsection{Simulations: Amazon Data}

We evaluate the LLM-based masking procedure on a semi-synthetic dataset built from real Amazon reviews, where the underlying text is natural language rather than a bag-of-words draw. We use reviews from the Health and Personal Care category of the Amazon Reviews corpus \citep{hou2024bridging}. To obtain a cleanly lexicon-identifiable treatment, we retain only reviews whose sentiment words come from {exclusively} a positive or a negative lexicon \citep{liu2010sentiment}, yielding $\approx 13,500$ reviews ($87\%$ positive). The unit of analysis is the product, with a single concatenated mega-review 
reflecting the reader-facing setting in which one evaluates the body of reviews for a single item. Treatment $T_i$ is the lexicon class of product $i$'s mega-review.

\textbf{Confounders and outcome generation. }
We extract latent confounders 
from the masked text representing the product categories and usage experience (details and list of  semantic labels are in Appendix \ref{app:sim:confounders}. Letting
$\bZ_{i} \in \{0,1\}^K$  indicate whether each covariate is present in document $i$. 

Thus, the confounding structure reflects real review content rather than a purely synthetic signal layered on top of the text. Outcomes are generated as
$y_i \sim f(\beta_0 + \beta_T T_i + \boldsymbol{\Gamma}^\top \bZ_i)$,
where $f$ is a logit link (binary) or Gaussian mean (continuous). Confounding strength is controlled by $\boldsymbol{\Gamma}$; we report a low and high setting for each outcome type.

\paragraph{Results.}

\begin{table}[h]\centering\small

\begin{tabular}{@{}>{}c l cccc cccc@{}}
\toprule
& & \multicolumn{4}{c}{\textbf{Low Confounding}} & \multicolumn{4}{c}{\textbf{High Confounding}} \\
\cmidrule(lr){3-6}\cmidrule(lr){7-10}
& Estimator & MSE & Bias & Var & PS & MSE & Bias & Var & PS \\
\midrule
\multirow{7}{*}{\rotatebox[origin=c]{90}{\textit{Binary }}}
& True & .0000 & .0001 & .0000 & -- & .0000 & .0006 & .0000 & -- \\
& Naive & .0230 & .1514 & .0001 & -- & .0553 & .2349 & .0001 & -- \\
\cmidrule(l){2-10}
& Masking  & .0038 & .0287 & .0030 & 96.0\% & .0115 & .1000 & .0015 & 97.7\% \\
& TextCause & .0237 & .1531 & .0002 & -- & .0565 & .2371 & .0003 & -- \\
& TI (Trimmed) & .0239 & .0509 & .0213 & 2.2\% & .0280 & .0491 & .0256 & 2.2\% \\
& TI (Winsorized) & .0363 & .1693 & .0076 & 2.2\% & .0436 & .1908 & .0072 & 2.2\% \\
& TI (Out. Reg.) & .0434 & .1887 & .0078 & 2.2\% & .0523 & .2124 & .0072 & 2.2\% \\

\midrule
\multirow{7}{*}{\rotatebox[origin=c]{90}{\textit{Continuous}}}
& True & .0229 & .0186 & .0226 & -- & .0229 & .0186 & .0226 & -- \\
& Naive & .3669 & .6058 & .0000 & -- & 1.1238 & 1.0601 & .0000 & -- \\
\cmidrule(l){2-10}
& Masking  & .0195 & -.0425 & .0177 & 87.0\% & .0442 & -.0674 & .0397 & 87.9\% \\
& TextCause & .2746 & .5219 & .0021 & -- & .6132 & .7779 & .0081 & -- \\
& TI (Trimmed) & .0336 & -.1072 & .0221 & 2.9\% & .0615 & -.1614 & .0355 & 5.5\% \\
& TI (Winsorized) & .0248 & .1420 & .0047 & 2.9\% & .0422 & .1626 & .0157 & 5.5\% \\
& TI (Out. Reg.) & .0380 & .1820 & .0049 & 2.9\% & .0497 & .1789 & .0177 & 5.5\% \\
\bottomrule
\end{tabular}
\caption{Binary (top) and continuous (bottom) outcomes on semi-synthetic Amazon reviews: MSE, bias, variance, and median fraction of held-out test documents with estimated propensity scores (PS) in $[0.1, 0.9]$ across 100 simulations for models with a PS head during training.\label{tab:results_adaptive_lambda}}
\end{table}

Table~\ref{tab:results_adaptive_lambda} compares the adaptive masking estimator against the naive difference-in-means estimator, TextCause, and TI-estimator across binary and continuous outcomes with low and high confounding. The \textbf{naive} estimator is the unadjusted differences in means between the treatment and control groups; it confirms that the outcomes are meaningfully confounded as its MSE increases with confounding strength. The line labeled ``\textbf{true}" calculates the average difference between the simulated 
treatment and control outcomes.
This retains noise from random outcome generation. 
The PS column reports the median fraction of held-out test documents with estimated propensity scores in $[0.1,0.9]$ across simulation repetitions.

\textbf{TextCause} closely tracks the naive estimator.
Without a treatment-prediction loss, the representation is not encouraged to retain treatment-predictive confounding information.  This is analogous to the LEACE failure mode in Section~\ref{sec:topic}. 
The \textbf{TI-estimator} suffers severe overlap violations: the median fraction of propensities in $[0.1, 0.9]$ ranges from 2–6\%, reflecting direct treatment encoding from the text. The propensity score distributions (Figure~\ref{fig:prop-grid}, Appendix~\ref{app:llm-comp}) show near-complete separation between treatment arms. We report trimmed and winsorized variants as the original AIPW produces degenerate estimates. The performance of TI is unstable; in some settings, MSE improves, while in others it remains comparable to or worse than the naive estimator. The \textbf{masking} estimator gives the strongest overall performance. The masking estimator achieves the lowest MSE in three of four settings and is competitive in the fourth (continuous high confounding). Replacement masking with adaptive overlap regularization retains confounding information without collapsing treated and control documents into separated regions. Ablation studies (Appendix~\ref{app:ab-mask-comp}) show that token masking and $L_{\textrm{penalty}}$ alone each provide improvements over competitors. This analysis of embedding similarity suggests that the penalty encourages the encoder to represent lexicon and [MASK] tokens similarly. 
Further ablation studies show that the mask-trained encoder evaluated on full text outperforms evaluation on masked text (Appendix~\ref{app:eval-mask}).

\section{Real Data Application}
We apply the LLM pipeline to complaints filed with the
Consumer Financial Protection Bureau (CFPB)\footnote{https://www.consumerfinance.gov/data-research/consumer-complaints/search/}, using the application data
of \citet{pryzant-etal-2021-causal} and
\citet{gui2023causalestimationtextdata}. The treatment is  hedging vs. confidence in a complaint; the outcome is if the company gave a timely response.

We restrict the CFPB database to complaints with non-empty narratives over a three-year window. Following \citet{mozer2020matching}, we pair each untimely complaint with its most similar timely complaint 
We keep the 4,000 highest-similarity pairs (8,000 complaints). Treatment $T_i$ is a lexicon-based confidence indicator built from the \texttt{politeness} package \citep{Yeomans2018}: complaints in the top quartile of hedge use relative to truth-intensifier use are treated; bottom-quartile are controls. Following \citet{pryzant-etal-2021-causal, gui2023causalestimationtextdata}, we include a covariate $C_i$ indicating if the complaint is about a Mortgage or Bank account/service. We mask hedge terms and truth intensifiers and fit the LLM from Section~\ref{sec:llm}, using 100 bootstrap train/validation/test splits for variance estimation and cross-validation for epoch selection.

\begin{wraptable}{r}{0.5\textwidth}
\vspace{-2em}
\centering
\small
\caption{CFPB effect of hedging on timely response.}
\label{tab:cfpb}
\begin{tabular}{lccr}
\toprule
Estimator & $\hat{\tau}$ & 95\% CI & PS \\
\midrule
Masking & $-$0.006 & ($-$0.035, 0.023) & 100\% \\
TextCause & $-$0.022 & ($-$0.046, $-$0.005) & -- \\
TI (Trim) & 0.003 & ($-$0.074, 0.075) & 71\% \\
TI (Winsor.) & $-$0.017 & ($-$0.096, 0.052) & 71\% \\
\bottomrule
\end{tabular}
\vspace{-1em}
\end{wraptable}

Table~\ref{tab:cfpb} summarizes the results. The naive difference in means shows a small negative effect of hedging ($\hat{\tau} = -0.033$). The masking estimator brings this near zero ($\hat{\tau} = -0.006$) with a confidence interval comfortably containing zero, and all propensity scores remain in $[0.1, 0.9]$. TextCause is the only estimator whose interval excludes zero; this apparent finding may reflect inadequate adjustment. The TI-estimator agrees directionally with masking but exhibits substantial overlap violations. The agreement between masking and TI on a null effect lends credibility to the conclusion that hedging has no effect on timely response. \textit{This suggests the interaction of two known opposing mechanisms. }Readers are known to engage more with less intellectually humble (more confident) writing \citep{katta2025llm}, while polite complaints have been found to receive faster responses \citep{pryzant-etal-2021-causal, gui2023causalestimationtextdata}.

\label{sec:real-data}
\section{Discussion and Future Work} 
\label{sec:discussion}

We identify a failure mode in causal inference with text: when treatment is encoded in the document, representations learned from the full text can induce overlap violations. We propose masking as a remedy that removes the lexical shortcut defining treatment that preserves information needed for adjustment. This distinguishes masking from approaches that erase all treatment-predictive information, which can restore overlap at the cost of discarding legitimate confounding signal.

\textbf{Limitations. } Several limitations point to future work. First, the strongest theoretical guarantee applies to deletion masking in bag-of-words settings; replacement masking for LLMs is necessarily weaker. Sharper diagnostics for residual treatment leakage would be valuable. Second, while they are found in many applications and are computationally simple, lexicon-based treatments are restrictive. Extending masking to phrase-level or model-discovered treatment definitions is an important next step, as is understanding how imperfect treatment dictionaries affect identification. Third, treatment-defining words may also carry adjustment-relevant information; our train-on-masked, evaluate-on-full design mitigates but does not eliminate this tension. Despite these limitations, the core message is clear: causal adjustment with text requires representations that encode enough non-treatment structure to control confounding without encoding the lexical signal that defines the treatment itself. Masking provides a simple, practical way to enforce this separation for lexicon-based treatments.

\bibliographystyle{plainnat}
\bibliography{custom}

\appendix

\section{Proof of Theorems}
\label{app:proofs}

\subsection{Theorem 1}
\label{app:th0}

\begin{proof}
By the tower property and Assumption~A1-A3,
\begin{align*}
    \mathbb{E}_B\!\bigl[\mathbb{E}(Y \mid T=t,\, B)\bigr]
    &= \mathbb{E}_B\!\bigl[\mathbb{E}(Y(t) \mid T=t,\, B)\bigr] & &\text{(Assumption~A1)} \\
    &= \mathbb{E}_B\!\bigl[\mathbb{E}(Y(t) \mid B)\bigr] & &\text{(Assumption~A2)} \\
    &= \mathbb{E}\!\bigl[Y(t)\bigr].
\end{align*}
Assumption A3 ensures that $\mathbb{E}[Y | T = t, B]$ is well-defined for $t \in \{0, 1\}$ so that the outer expectation exists.
The result follows by taking $t = 1$ and $t = 0$ and subtracting.
\end{proof}

\subsection{Theorem 2}
\label{app:th1}

\begin{proof}
Let $A_{\alpha} = \{\mathbf{w} : \alpha \leq e_b(\mathbf{w}) \leq 1 - \alpha\}$. On $A_{\alpha}$, $\min(e_b(\mathbf{W}), 1 - e_b(\mathbf{W})) \geq \alpha$. Thus, we have 
\begin{align*}
    R_b & = \mathbb{E}_{\mathbf{W}}\!\big[\min\!\big(\Pr(T{=}1 \mid b(\mathbf{W})),\; 1 - \Pr(T{=}1 \mid b(\mathbf{W}))\big)\big] \\
    & = \mathbb{E}_{\mathbf{W}}\!\big[\min\!\big(e_b(\mathbf{W}), 1 - e_b(\mathbf{W})\big)\big] \\
    & \geq \mathbb{E}_{\mathbf{W}}\!\big[\min\!\big(e_b(\mathbf{W}), 1 - e_b(\mathbf{W})\big) \cdot \mathbf{1}(\mathbf{W} \in A_{\alpha})\big] \geq \alpha \cdot \Pr(\mathbf{W} \in A_\alpha). 
\end{align*}
\end{proof}

\subsection{Theorem 4}
\label{app:th3}
\begin{proof}
{Since $b_{-}(\mathbf{W})$ is a function of $\mathbf{W}^{-}$, the assumption
$T \perp \mathbf{W}^{-} \mid \mathbf{Z}$ implies
$T \perp b_{-}(\mathbf{W}) \mid \mathbf{Z}$, so
$$
\Pr(T{=}1 \mid \mathbf{Z},\, b_{-}(\mathbf{W}))
\;=\;
\Pr(T{=}1 \mid \mathbf{Z})
\;=\;
e^{\star}(\mathbf{Z}).
$$
By the tower property,
\begin{align*}
e_{b_{-}}(\mathbf{W})
&=\; \Pr(T{=}1 \mid b_{-}(\mathbf{W}))\\
&=\; \mathbb{E}_{\mathbf{Z}}\!\left[
    \mathbb{E}\!\left[
        T \,\big|\, \mathbf{Z},\, b_{-}(\mathbf{W})
    \right]
    \,\big|\,
    b_{-}(\mathbf{W})
\right] \\
&=\; \mathbb{E}_{\mathbf{Z}}\!\left[
    e^{\star}(\mathbf{Z})
    \,\big|\,
    b_{-}(\mathbf{W})
\right].
\end{align*}
By the overlap assumption,
$e^{\star}(\mathbf{Z}) \in [\alpha^{\star}, 1-\alpha^{\star}]$,
so
$e_{b_{-}}(\mathbf{W}) \in [\alpha^{\star}, 1-\alpha^{\star}]$.
}

\end{proof}

\subsection{Proposition 5}
\label{app:p1}

\noindent\emph{Construction.} Let $\mathbf{Z}, T \sim \mathrm{Bernoulli}(1/2)$ independently, so latent overlap holds. Suppose treated documents contain exactly one lexicon token and controls contain none. Let $b_M$ be a mask-trained encoder that is sensitive to \texttt{[MASK]} tokens (encountered during training) but invariant to lexicon tokens (never encountered during training). On masked text, $b_M(\mathbf{W}^{\mathrm{mask}})$ recovers the mask indicator, which equals $T$ in this DGP, so $\Pr(T = 1 \mid b_M(\mathbf{W}^{\mathrm{mask}})) \in \{0,1\}$ and overlap fails. On full text, $b_M(\mathbf{W})$ depends only on non-lexicon content, hence is a function of $\mathbf{Z}$, giving $\Pr(T=1 \mid b_M(\mathbf{W})) = 1/2$ and overlap holds. \qed

This construction provides intuition for why masked evaluation may not be safer than full-text evaluation: when the mask pattern itself provides strong information on $T$, as may be the case in real text, evaluating on $\mathbf{W}$ can be preferable. Full-text evaluation is further appealing for LLMs as lexicon words may carry adjustment-relevant outcome information once the direct path to treatment is removed during training.

\subsection{A Sensitivity Bound for Replacement Masking}
\label{app:replacement-bound}

Section \ref{sec:masking} introduces replacement masking as a controlled relaxation of deletion masking but does not formalize the sense in which residual leakage degrades overlap rather than destroying it. This appendix states and proves a sensitivity bound that makes the relaxation precise.

Recall the mask-trained representation $b_M(\mathbf{W}) := b_{\theta_M}(\mathbf{W})$ and its induced propensity score $e_M(\mathbf{W}) := \Pr(T = 1 \mid b_M(\mathbf{W}))$. To compare $e_M$ to the deletion-safe case from Theorem \ref{prop:masking}, we benchmark against the deletion-masked propensity score evaluated on documents that share a common mask-trained representation:
$$
\widetilde{e}_M^{\mathrm{del}}(\mathbf{W}) := \mathbb{E}_{\mathbf{W}'}\!\left[e_{b_{\mathrm{del}}}(\mathbf{W}') \,\big|\, b_M(\mathbf{W}') = b_M(\mathbf{W})\right].
$$
This is the average deletion-masked propensity score among documents that share the same mask-trained representation value. Because $e_{b_{\mathrm{del}}}(\mathbf{W}) \in [\alpha^{\star}, 1 - \alpha^{\star}]$ for every document by Theorem \ref{prop:masking}, the conditional expectation preserves the bounds:
$$
\widetilde{e}_M^{\mathrm{del}}(\mathbf{W}) \in [\alpha^{\star}, 1 - \alpha^{\star}].
$$
Intuitively, $\widetilde{e}_M^{\mathrm{del}}(\mathbf{W})$ is the propensity one would expect for a document with mask-trained representation $b_M(\mathbf{W})$ if the only treatment-relevant information in that representation came through confounding. Deviations of $e_M(\mathbf{W})$ from $\widetilde{e}_M^{\mathrm{del}}(\mathbf{W})$ measure treatment information present in the mask-trained representation but absent from the deletion-masked representation. We call the average absolute deviation the \emph{treatment-information surplus}:
$$
\epsilon_M := \mathbb{E}_{\mathbf{W}}\!\left[\left| e_M(\mathbf{W}) - \widetilde{e}_M^{\mathrm{del}}(\mathbf{W}) \right|\right].
$$
The quantity $\epsilon_M$ is zero exactly when the mask-trained representation carries no treatment information beyond what the deletion-masked representation provides, and increases as residual leakage grows.

\begin{theorem}[Sensitivity bound for replacement masking]
\label{thm:replacement-app}
For any $\eta > 0$,
$$
\Pr\!\left(e_M(\mathbf{W}) \notin [\alpha^{\star} - \eta,\ 1 - \alpha^{\star} + \eta]\right) \;\leq\; \frac{\epsilon_M}{\eta}.
$$
\end{theorem}

\begin{proof}
Since $\widetilde{e}_M^{\mathrm{del}}(\mathbf{W}) \in [\alpha^{\star}, 1 - \alpha^{\star}] \subseteq [\alpha^{\star} - \eta,\ 1 - \alpha^{\star} + \eta]$, any document with $e_M(\mathbf{W})$ outside $[\alpha^{\star} - \eta,\ 1 - \alpha^{\star} + \eta]$ must satisfy $\lvert e_M(\mathbf{W}) - \widetilde{e}_M^{\mathrm{del}}(\mathbf{W}) \rvert > \eta$. Therefore
$$
\bigl\{ \mathbf{W} : e_M(\mathbf{W}) \notin [\alpha^{\star} - \eta,\ 1 - \alpha^{\star} + \eta] \bigr\}
\;\subseteq\;
\bigl\{ \mathbf{W} : \lvert e_M(\mathbf{W}) - \widetilde{e}_M^{\mathrm{del}}(\mathbf{W}) \rvert > \eta \bigr\},
$$
so
$$
\Pr\!\left(e_M(\mathbf{W}) \notin [\alpha^{\star} - \eta,\ 1 - \alpha^{\star} + \eta]\right)
\;\leq\;
\Pr\!\left(\lvert e_M(\mathbf{W}) - \widetilde{e}_M^{\mathrm{del}}(\mathbf{W}) \rvert \geq \eta\right).
$$
Markov's inequality applied to the nonnegative random variable $\lvert e_M(\mathbf{W}) - \widetilde{e}_M^{\mathrm{del}}(\mathbf{W}) \rvert$ gives
$$
\Pr\!\left(\lvert e_M(\mathbf{W}) - \widetilde{e}_M^{\mathrm{del}}(\mathbf{W}) \rvert \geq \eta\right)
\;\leq\;
\frac{\mathbb{E}_{\mathbf{W}}\!\left[\lvert e_M(\mathbf{W}) - \widetilde{e}_M^{\mathrm{del}}(\mathbf{W}) \rvert\right]}{\eta}
\;=\;
\frac{\epsilon_M}{\eta}.
$$
\end{proof}

\paragraph{Interpretation and scope.}
Theorem \ref{thm:replacement-app} should be read as a sensitivity bound rather than a usable identification result. The treatment-information surplus $\epsilon_M$ is defined relative to an unobservable benchmark $\widetilde{e}_M^{\mathrm{del}}$ and is not directly estimable from data; the safe-overlap margin $\alpha^{\star}$ is also unknown. What the bound delivers is a structural statement: as long as $\epsilon_M$ is small relative to a chosen tolerance $\eta$, the fraction of documents whose propensity scores can drift outside the band $[\alpha^{\star} - \eta,\ 1 - \alpha^{\star} + \eta]$ is correspondingly small. Replacement masking can therefore degrade overlap smoothly with leakage, rather than collapsing as in the unmasked case identified by Theorem \ref{thm:approx}.

Because $\epsilon_M$ is not estimable, we treat the empirical fraction of estimated propensity scores in $[\alpha,\ 1 - \alpha]$ as the practical diagnostic. Sections \ref{sec:topic} and \ref{sec:llm} report this diagnostic for every simulation. A high in-band fraction is consistent with small $\epsilon_M$; a low fraction signals that the mask-trained representation has retained enough residual treatment information to compromise adjustment, regardless of what the formal bound would assert.

$L_{\text{penalty}}$ is the training-time counterpart to the bound here, capping how often the trained representation produces propensity scores outside the overlap region.

\paragraph{Relation to deletion masking.}
At $\epsilon_M = 0$, every document satisfies $e_M(\mathbf{W}) = \widetilde{e}_M^{\mathrm{del}}(\mathbf{W}) \in [\alpha^{\star}, 1 - \alpha^{\star}]$ almost surely, recovering Theorem \ref{prop:masking} as the limiting case of zero residual leakage. Theorem \ref{thm:replacement-app} thus situates deletion masking as the boundary point of a one-parameter family in which replacement masking trades exact overlap preservation for richer adjustment information at evaluation.

\section{Potential Outcomes and Text as Treatment}
\label{sec:appendix-po}

\subsection{Causal Text Estimands}
\label{app:textest}

In Section \ref{sec:CB}, we introduced $\tau_t$, the treatment effect of the latent treatment as experienced by the reader. However, it is also possible to investigate $\tau_d$, the treatment effect of the known document label.

Before expanding on the estimands, we first want to emphasize the difference in reader vs. writer perspective shown in Figure \ref{fig:dag}. This is a distinction unique to text-based causal inference and is line with literary theory separating language into perspectives of author intent and reader perception \citep{iser1978act, harker1988lit}. While the writer develops the text, we are interested in some eventual outcome from the reader. 

With this perspective in mind, we now present details of and differences between $\tau_d$ and $\tau_t$; refer to \citet{katta2025llm} for an even more thorough development of these ideas. 
To clearly understand the difference between these two effects, we highlight that $\tau_d$ can be thought of as the effect of exposure to \textit{documents} with some latent feature. We are intervening on \textit{D} and seeing the impact of the outcome on $Y$. In potential outcomes notation, $\tau_D$ is given by 

\begin{align*}
    \tau_d &= \mathbb{E}[Y_i(D=1) - Y_i(D=0)] \\
    & = \mathbb{E}[Y_i(T=1,Z=z) - Y_i(T=0,Z=z')].
\end{align*}

As $D$ varies from $0$ to $1$, we assume that $T$ does as well. It is also possible (and likely) for $Z$ to vary between some $z$ and $z'$. Importantly, here, $Z$ is not held fixed as we take the expectation over $i$ (units in the population). As a result, we do not need to adjust for confounding by $Z$ when computing $\tau_D$. To estimate $\tau_D$, it is enough to perform a completely randomized experiment: gather collections of documents with $D = 1$ and $D = 0$. Assign individuals randomly to one document (naturally either treatment or control) and measure outcomes; a difference in means suffices to estimate the treatment effect. 

The same process does not apply to estimating $\tau_t$. Instead of measuring the effect of \textit{documents} with a feature, $\tau_t$ isolates the effect of the \textit{latent treatment itself} and is given by 
\begin{align*}
    \tau_t &= \mathbb{E}[Y_i(T = 1) - Y_i(T = 0)] \\&= \mathbb{E}[Y_i(T=1,Z=z) - Y_i(T=0,Z=z)].
\end{align*}
That is, $\tau_t$ measures how changing the value of $T$ changes the outcome, holding $Z = z$ fixed. In an ideal world, we would change $T$ without changing $W$ (and hence without changing $Z$), then use a simple randomized experiment. However, changing $T$ requires changing the documents' words $W$ and thus possibly changing $Z$, leading to complications in estimating $\tau_t$. 

One approach taken in existing works, such as in \citet{fong2023causal} and \citep{katta2025llm} is to construct or curate texts such that $T$ is believed from domain knowledge and/or empirical validation to be independent of $Z$. While successful curation of such documents does allow for identification of $\tau_t$, this is a non-trivial task and strongly limits the ability to use existing documents in an observational setting. Another approach, such as in \citet{gui2022causal}, is the one taken in this paper: learning and adjusting for confounders $Z$. 

Regardless of the estimation approach, understanding how to disentangle $T$ and $Z$ for a given document is the core of the challenge in measuring effects of latent features. 

\subsection{Full Form of Estimators}
\label{app:estimators}

Throughout the paper we reference three estimators of the ATE: outcome
regression (OR), inverse propensity weighting (IPW), and the augmented
inverse propensity weighting estimator (AIPW). We give their full forms
here for reference.

\paragraph{Outcome regression.} Let $\mu_t(z) = \mathbb{E}[Y(t) \mid Z = z]$
for $t \in \{0, 1\}$. Given estimators $\hat{\mu}_0$ and $\hat{\mu}_1$, the OR estimator is
$$
\hat{\tau}^{\text{reg}} \;=\; \frac{1}{n}\sum_{i=1}^{n}
\big(\hat{\mu}_1(Z_i) - \hat{\mu}_0(Z_i)\big).
$$
If treated units are absent for some value of $z$, $\hat{\mu}_1(z)$ is not
identifiable from the data; the analogous failure for $\hat{\mu}_0$ occurs
when control units are absent. Overlap rules both out.

\paragraph{Inverse propensity weighting.} Let
$e(Z) = \Pr(T = 1 \mid Z)$ with estimator $\hat{e}$. The H\'{a}jek IPW
estimator \citep{hajek1971comment} is
$$
\hat{\tau}^{\text{hajek}} \;=\;
\frac{ \sum_{i} T_i Y_i / \hat{e}(Z_i) }
     { \sum_{i} T_i / \hat{e}(Z_i) }
\;-\;
\frac{ \sum_{i} (1 - T_i) Y_i / (1 - \hat{e}(Z_i)) }
     { \sum_{i} (1 - T_i) / (1 - \hat{e}(Z_i)) }.
$$
When overlap fails, $\hat{e}(Z_i)$ is near $0$ or $1$ for some units and
the inverse-weighted terms blow up.

\paragraph{Augmented IPW.} Combining the two
\citep{robins1995aipw}:
$$
\hat{\tau}^{\text{aipw}} \;=\; \frac{1}{n}\sum_{i=1}^{n}
\left[ \hat{\mu}_1(Z_i) - \hat{\mu}_0(Z_i)
+ \frac{T_i\,(Y_i - \hat{\mu}_1(Z_i))}{\hat{e}(Z_i)}
- \frac{(1 - T_i)\,(Y_i - \hat{\mu}_0(Z_i))}{1 - \hat{e}(Z_i)} \right].
$$
$\hat{\tau}^{\text{aipw}}$ is consistent if either the outcome model or the propensity
score model is correctly specified.

\section{Appendix: LDA Implementation and Simulations}
\label{app:all_lda}

\subsection{LDA Plate Diagram}
\label{app:LDA}

The LDA-based document structure is illustrated in Figure \ref{fig:lda}. 

\begin{figure}[h]
\centering
\includegraphics[width=0.7\columnwidth]{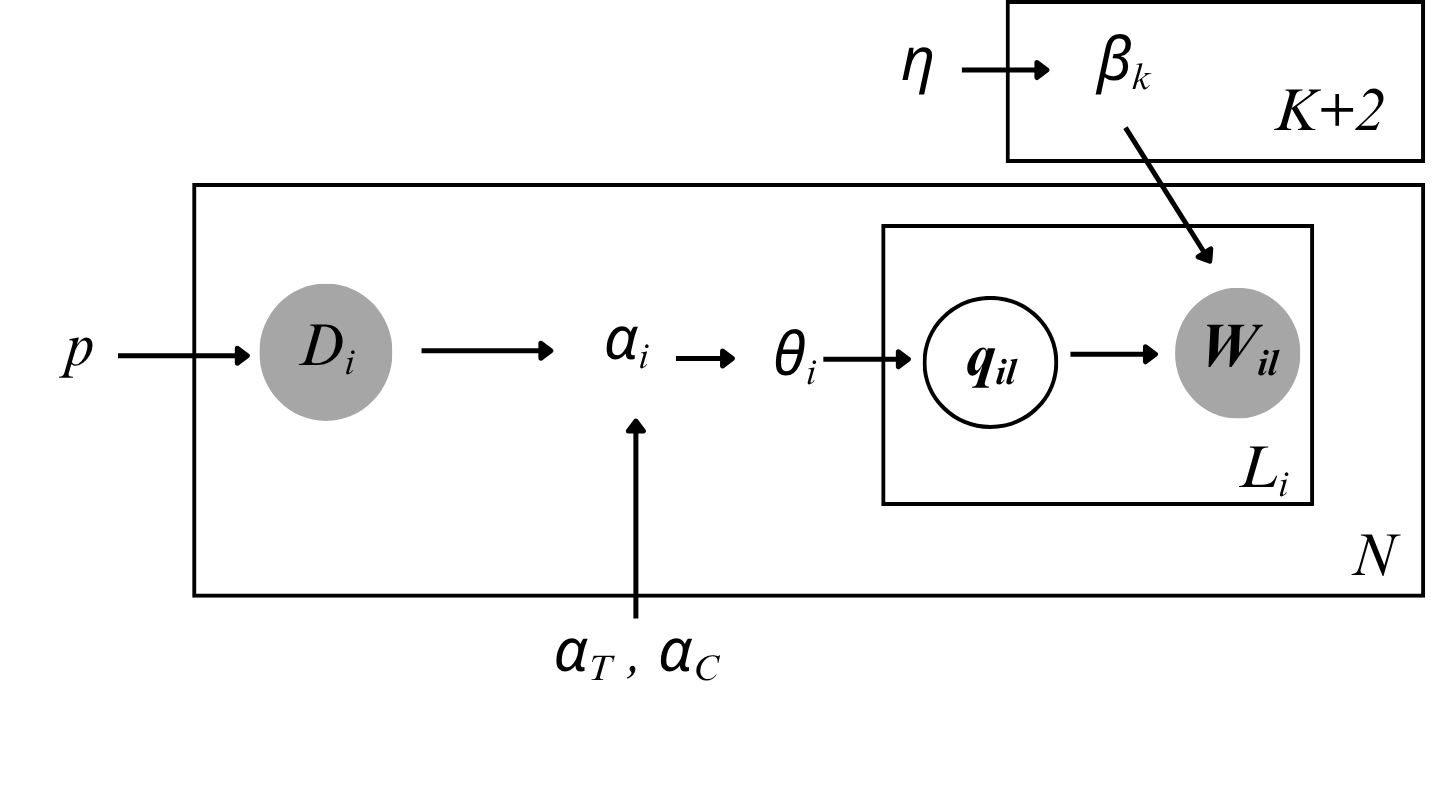}
  \caption{LDA plate diagram for assumed document structure.}
  \label{fig:lda}
\end{figure}

For a given document $i \in N$, the length $L_i$ (number of words) and true document label $D_i$ are drawn 
\begin{align*}
    L_i &\sim \mathcal{L} \\
    D_i &\sim Bern(p)
\end{align*}
The topic proportion parameter $\alpha_i$ is $\alpha_T$ when the document is treated and $\alpha_C$ when it is control. That is, $\alpha_i = D_i\alpha_T + (1-D_i)\alpha_C$. 
From this parameter, document-level topic proportion vector $\theta_i$ is drawn $$\theta_i \sim Dirichlet(\alpha_i).$$
Words $l \in \{1 \ldots L_i \} $ are drawn by first drawing topic $$q_{il} \sim Multinomial(\theta_i)$$ then word $$W_{il} \sim Multinomial(\beta_{q_{il}})$$ 
where $\beta_k$ for topic $k \in \{1 \ldots K + 2\}$ is a topic-specific word probability vector. With a dictionary of size $M$, this lies in the $M - 1$ simplex. Generally, we assume $$\beta_k \sim Dirichlet(\eta).$$

\subsection{Implementation Details}
\label{app:lda-implement}

Propensity scores are computed via logistic regression and winsorized to $[0.1, 0.9]$ to mitigate the impact of extreme values. We use an 80/20 train/test split: all models are fit on the training split and the ATE is estimated on the held-out test set to avoid post-selection bias \citep{egami2018makecausalinferencesusing}.

To identify treatment-defining words for masking, we first fit LDA with $K+2$ topics to the full (unmasked) corpus. We then compare the resulting topics to the known treatment and control word lists used in data generation, selecting the two topics with the highest overlap. The top words from these two topics form the masking lexicon. After removing these words from each document's bag-of-words representation, we refit LDA with $K$ topics on the masked text. The resulting $K$-dimensional topic proportions serve as the adjustment representation.

All LDA models are fit using the \texttt{lda} function from the R \texttt{topicmodels} package (GPL-2 license) with the VEM algorithm \citep{grun2011tm}.

\subsection{Associated Press Parameters}
\label{app:AP}

The parameters for simulations in section \ref{sec:sims} are based on an analysis of the Associated Press data in the R \texttt{topicmodels} package. As the aim of this aspect of our work is not to study underlying true causal effects in data, but rather to use the data as a base for choosing parameters, we do not require all parameters to be the underlying \textit{true} fits to the data. 

This data originally contains bag of words representations of 10,473 terms over 2,246 documents. Before continuing with parameter selection, we compute tf-idf to filter for words with the highest scores \citep{blei2009topic}, resulting in a dictionary of size 7,854. Taking $K + 2 = 20$, we fit LDA to the data and examine posterior within-document Dirichlet parameter $\alpha$ and within-topic word distribution parameter $\beta$.
To choose treatment and control topics, we examine each topics' top 100 words, selecting topics with minimal cross-topic overlap. Using these selections, we modify $\beta$:
\begin{itemize}
    \item The top 100 words are selected within treatment and control topics. In each, all other words' probabilities are set to 0.
    \item The probabilities of these 200 words are then set to 0 in all other topics. 
    \item All within-topic word probabilities are rescaled to sum to 1. 
\end{itemize}
Additionally, we modify $\alpha$ to form two parameters: $\alpha_T$ and $\alpha_C$. 
To keep documents distinctly treatment or control, we set $\alpha_T[K+1] = \alpha_C[K+2] = 0$. This means treatment documents never sample the control topic and vice versa. Remaining entries are chosen such that different confounding topics are more likely in treatment and control documents, inducing confounding. Both $\alpha_T$ and $\alpha_C$ sum to 0.5 (as did $\alpha$), keeping documents primarily composed of a smaller subset of the topics instead of being more evenly dispersed.  Values for confounder topics in $\alpha_T$ and $\alpha_C$ fall between 0.01 and 0.035; $\alpha_T[K+1]$ and $\alpha_C[K+2]$ are larger to result in treatment and control topics often being slightly more prominent in the text; this models the idea that treatment and control topics might be the focus of the documents. Figure \ref{fig:props} shows an illustrative example after one complete document generation run with $N = 2500$ of distributions of true within-document topic proportions for each of 20 topics. 
\begin{figure}
    \centering
    \includegraphics[width=0.75\linewidth]{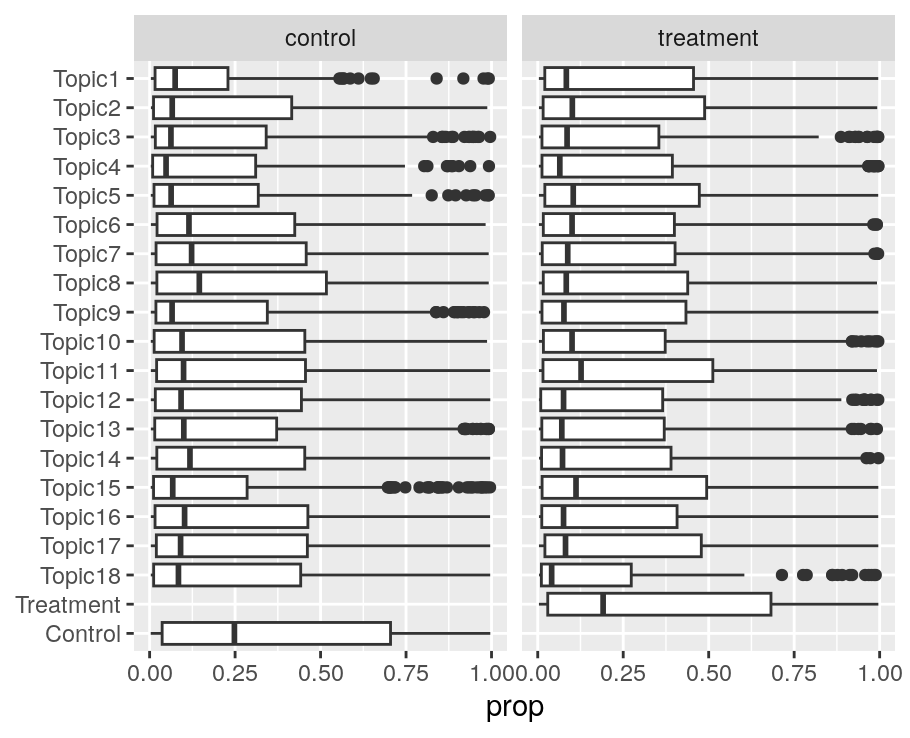}
    \caption{For each topic (y axis), the distributions of their true proportions across 2500 generated documents.}
    \label{fig:props}
\end{figure}

In outcome generation, confounding coefficients $\boldsymbol{\gamma} = [1, \ldots, \frac{2}{K}, \frac{1}{K}, 0, 0]$ where the last two entries correspond to treatment/control topics. 
\subsection{Computational Setup}
\label{app:lda-comp}
Results in Section~\ref{sec:topic} (and related appendices) were generated using a HPC with 60 cores each assigned 1GB of RAM. The processor was 
Intel(R) Xeon(R) Gold 6336Y. However, this is not a computationally intensive pipeline and running on a laptop is feasible. 

\subsection{LDA Simulations: Additional Figures}
\label{app:extra_lda}

We provide two additional presentations of results from Section~\ref{sec:sims}. 

Figure \ref{fig:rmse-app} compares the squared error values for outcome regression, IPW, and AIPW estimators computed using the full $K+2$ topics (x-axis) and masked $K$ topics (y-axis). Points that fall below the $x = y$ line have a higher squared error for the \textbf{*\_full} estimators. In outcome regression and AIPW, our masked estimators outperform the full estimator 64 and 62\% of the time, respectively. This is particularly impressive as the full covariates align with the underlying document structure and the outcome model was properly specified. For the Hájek estimator, this rises to 83\%.

\begin{figure}[h]
    \centering
    \includegraphics[width=0.5\linewidth]{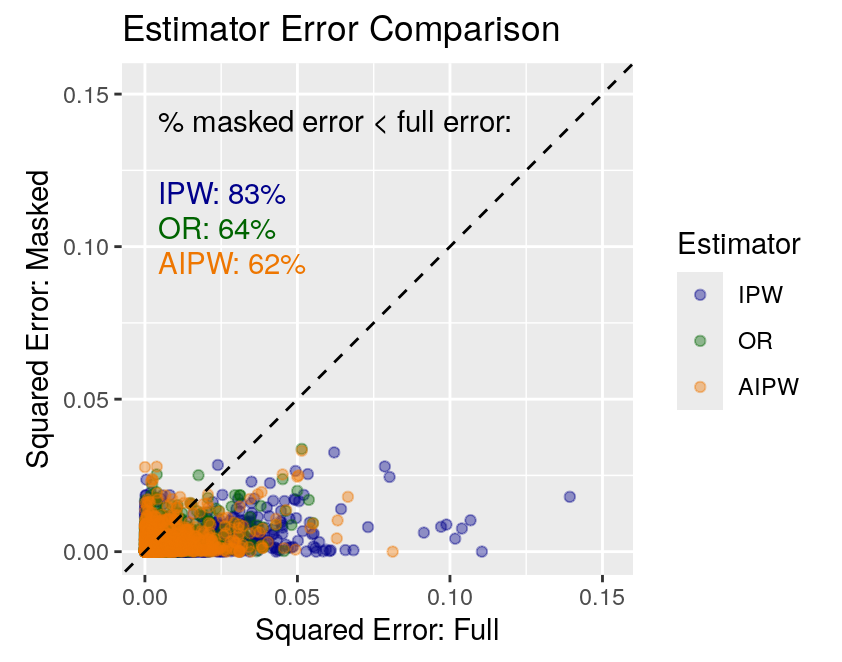}
    \caption{Squared error of \textbf{*\_full} (x-axis) vs. \textbf{*\_masked} (y-axis) AIPW (orange), Hajek IPW (blue), and outcome regression (green) estimates. Points below dashed $x = y$ line have lower error with masking. }
    \label{fig:rmse-app}
\end{figure}

Table \ref{tab:overlap-app} summarizes the proportions of documents with such values across iterations. In most masking iterations, there are nearly zero extreme values, whereas in most iterations with the full $K+2$ topics, \textit{all} documents have extreme propensity scores.  

\begin{table}[h]
\centering
\small
\caption{Distribution of share of documents with propensity scores outside $[0.1, 0.9]$ across simulations.}
\label{tab:overlap-app}
\begin{tabular}{@{}l cccccc@{}}
\toprule
Method & Min & Q1 & Median & Mean & Q3 & Max \\
\midrule
full   & .962 & 1.000 & 1.000 & .9998 & 1.000 & 1.000 \\
mask   & .000 & .000  & .006  & .0172 & .030  & .244 \\
leace  & .000 & .000  & .000  & .0002 & .000  & .224 \\
\bottomrule
\end{tabular}
\end{table}

\subsection{LDA Sensitivity Analysis: Topic Misspecification}
\label{app:ldamispec}

There are many methods for determining the correct number of topics to use in LDA \citep{JMLR:v25:23-0188, blei2003lda}. The details of each are beyond the scope of this work, however, choosing $K$ is an important aspect of any topic modeling problem. Despite the number of methods to choose an optimal value, misspecification from the latent truth is likely. 
Here, we investigate how sensitive our method is to an analyst's proposed value of $K+ 2$: \textit{what happens if we choose the wrong number of topics when running LDA to estimate confounding and find the treatment/control topic words?}. 

Using the same parameters as in Section \ref{sec:sims}, we take the value of $K+2$ used in LDA to estimate confounders to be $K + 2 \in \{10, 15, 18, 22, 25, 30\}$, showing varying levels of incorrect guesses on the number of topics. We emphasize that in all cases, the true underlying data generation remains with $K + 2 = 20$.  Figure \ref{fig:mispec-grid} show the squared error plots comparing the full vs. masked models over 1000 iterations. We see that in all cases, the masking procedure \textit{\textbf{maintains improvement}}. We see this most notably in the $K + 2 = 10$ case. While both masking and the full method have larger squared error values due to misspecification, they still improve over non-masked methods.

\begin{figure*}[t]
  \centering
  \begin{subfigure}[b]{0.49\textwidth}
    \centering
    \includegraphics[width=\linewidth]{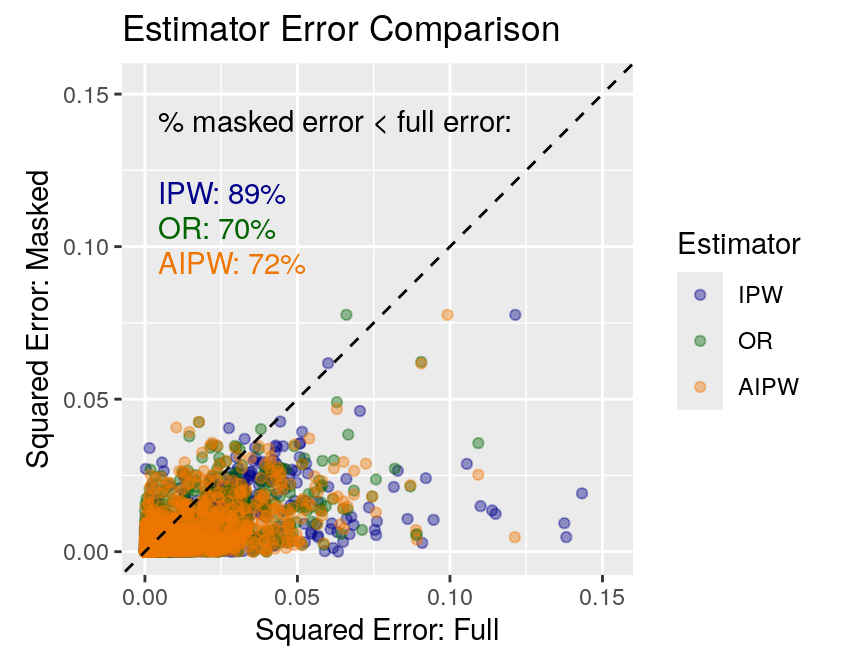}
    \caption{$K + 2 = 10$}
  \end{subfigure}\hfill
  \begin{subfigure}[b]{0.49\textwidth}
    \centering
    \includegraphics[width=\linewidth]{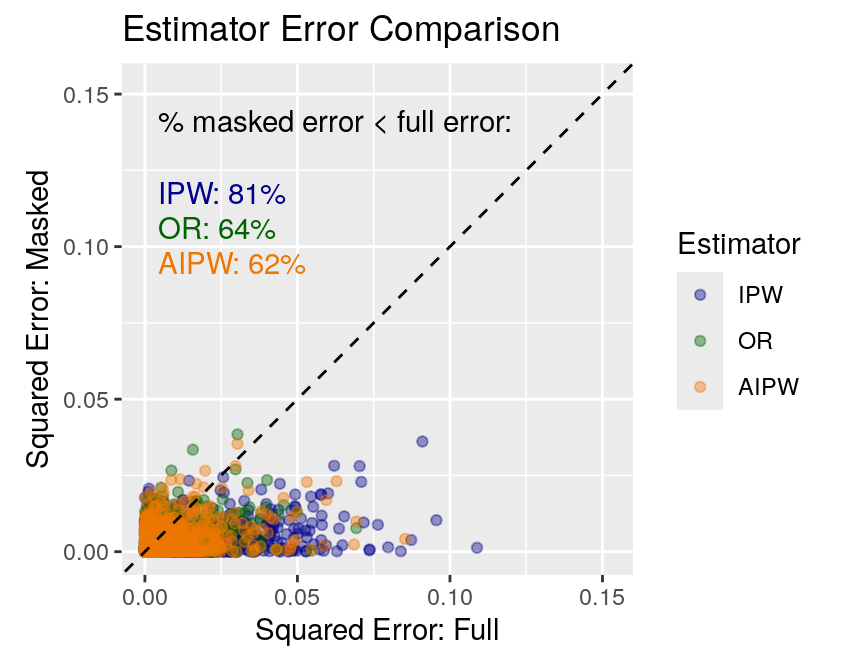}
    \caption{$K + 2 = 15$}
  \end{subfigure}

  \vspace{1ex}

  \begin{subfigure}[b]{0.49\textwidth}
    \centering
    \includegraphics[width=\linewidth]{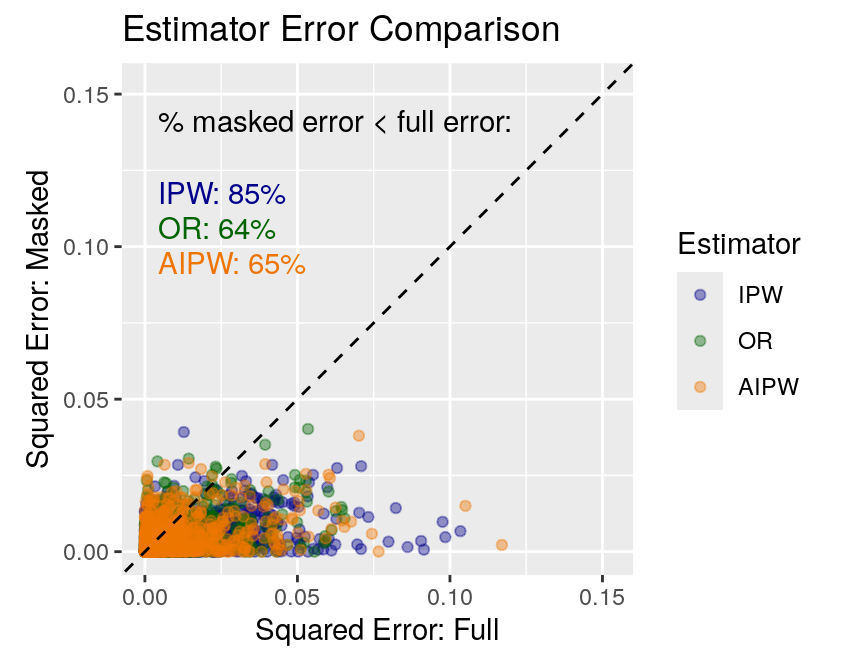}
    \caption{$K + 2 = 18$}
  \end{subfigure}\hfill
  \begin{subfigure}[b]{0.49\textwidth}
    \centering
    \includegraphics[width=\linewidth]{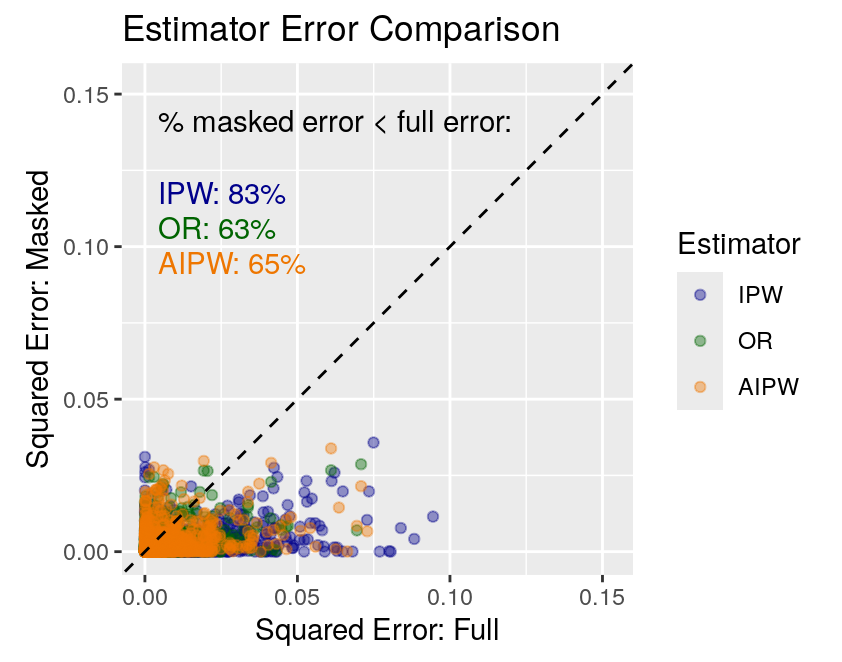}
    \caption{$K + 2 = 22$}
  \end{subfigure}

  \vspace{1ex}

  \begin{subfigure}[b]{0.49\textwidth}
    \centering
    \includegraphics[width=\linewidth]{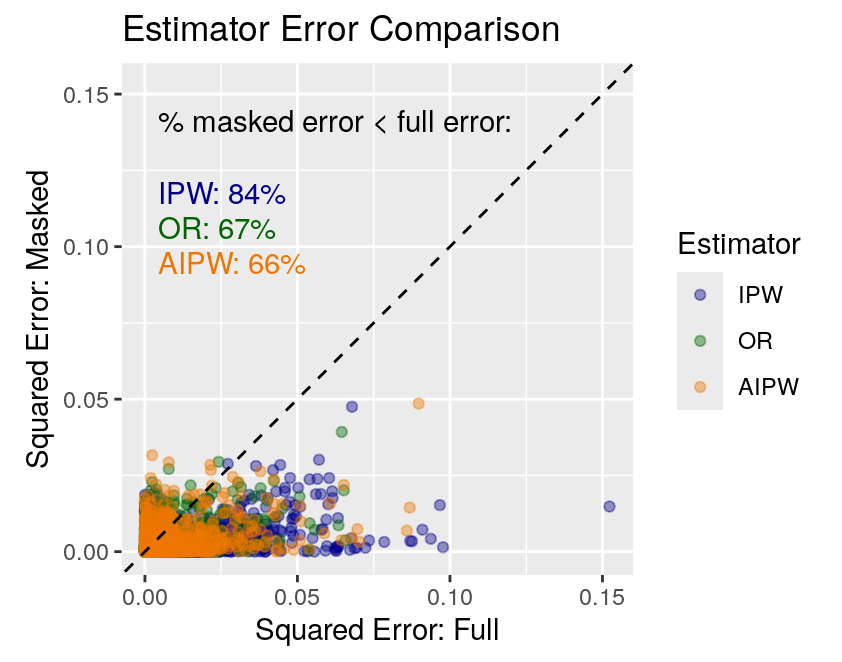}
    \caption{$K + 2 = 25$}
  \end{subfigure}\hfill
   \begin{subfigure}[b]{0.49\textwidth}
    \centering
    \includegraphics[width=\linewidth]{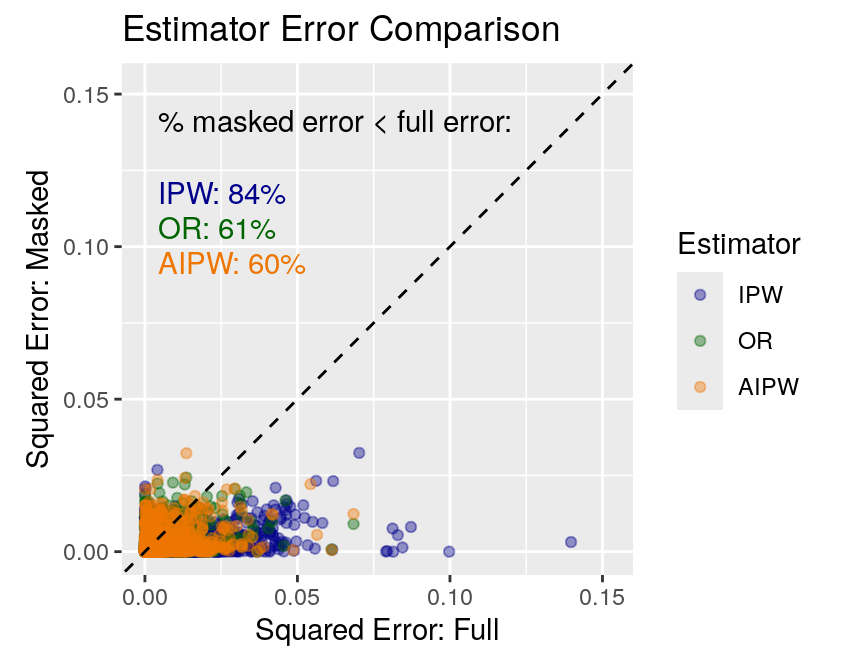}
    \caption{$K + 2 = 30$}
  \end{subfigure}

  \caption{Mispecification plots for topics 10, 15, 18, 22, and 25.}
  \label{fig:mispec-grid}
\end{figure*}

\subsection{LDA Results for Varied Simulation Parameters}
\label{app:moreparamlda}

We present illustrative results from varying document-generation parameters. Specifically, we vary the word probability parameters $\beta$ and the sum of $\alpha_T$ and $\alpha_C$ (denoted $\sum\alpha$). 

Here we show results for $K + 2 = 30$ and the same dictionary size as above (7854). We take $\sum\alpha\in \{ 0.5 , 5\}$. When $\sum\alpha = 0.5$, documents are generally dominated by a small set of topics. When $\sum\alpha = 5$, documents contain a wider variety of topics (though still not evenly dispersed).  To vary $\beta$, we modify Dirichlet parameter $\eta$. We take take $\eta_t \in \{ 0.1, 0.5\}$ as the Dirichlet parameter for within-topic word distributions of treatment and control topics. These topics have each have probabilities distributed across 100 non-overlapping words. For remaining topics, we take $\eta_{t} \in \{ 0.1, 0.5\}$ as the Dirichlet parameter for the within-topic distributions of remaining non-treatment, non-control words.  We refer to $\beta_1$ as the word probability parameters drawn when $\eta_t = 0.1, \eta_z = 0.01$ and $\beta_5$ as the word probability parameters drawn when $\eta_t = 0.5, \eta_z = 0.05$. $\beta_1$ results in topics dominated by fewer words while $\beta_5$ spreads words more evenly throughout topics.

Figure \ref{fig:more} shows results for each combination of $\beta$ and $\sum\alpha$ over 300 iterations each. In all but one case, the masking pipeline improves the error of the AIPW and outcome regression estimators. However, using $\beta_5$, $\sum\alpha = 5$, the full model has extremely small error that is not improved by masking. In all cases, the masking pipeline again improves the error of the Hájek IPW estimators; this is most notable in the $\beta_1$, $\sum\alpha = 5$ case, where the IPW estimator has lower error 97\% of the time. 

\begin{figure*}[t]
  \centering
  \begin{subfigure}[b]{0.49\textwidth}
    \centering
    \includegraphics[width=\linewidth]{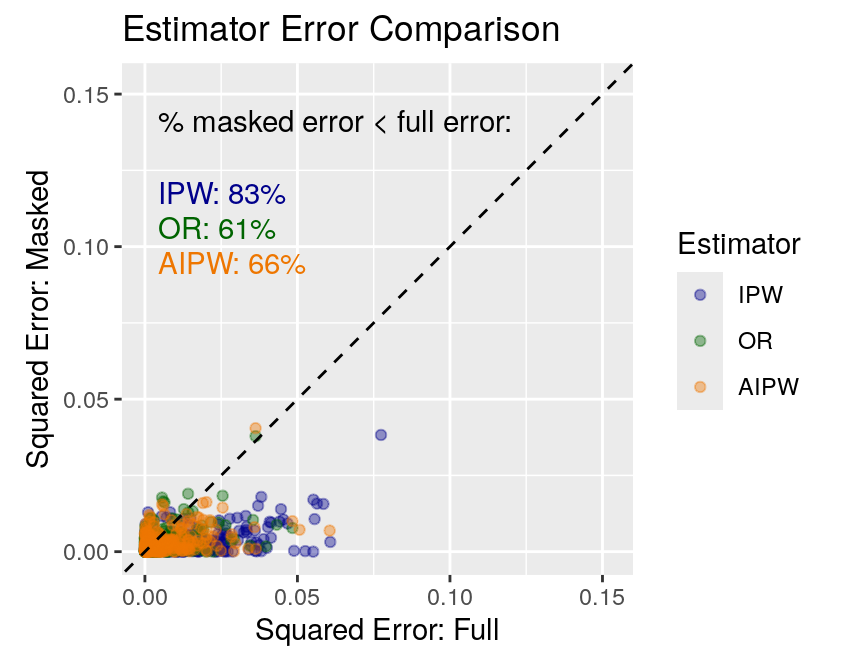}
    \caption{$\beta_1, \sum\alpha = 0.5$}
  \end{subfigure}\hfill
  \begin{subfigure}[b]{0.49\textwidth}
    \centering
    \includegraphics[width=\linewidth]{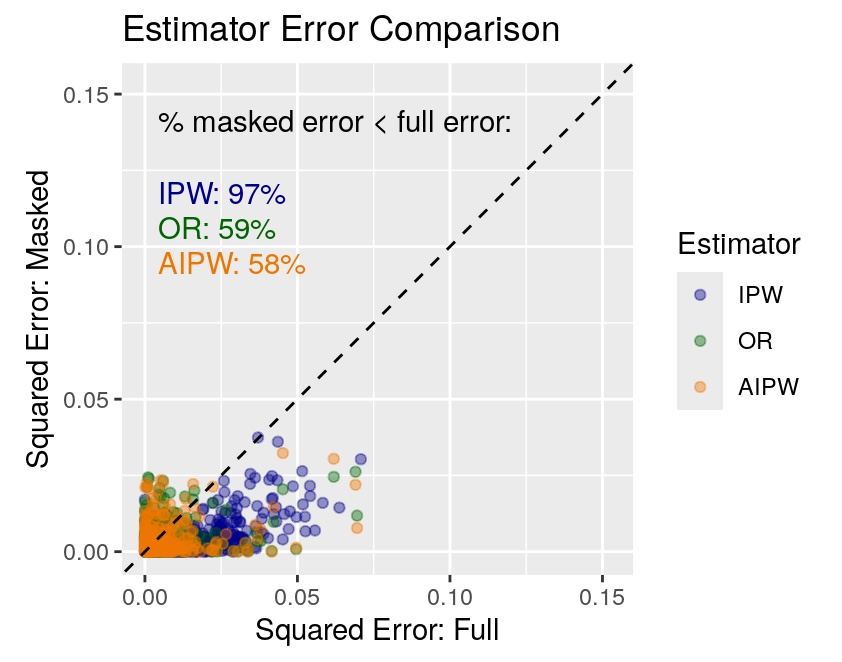}
    \caption{$\beta_1, \sum\alpha = 5$}
  \end{subfigure}

  \vspace{1ex}

  \begin{subfigure}[b]{0.49\textwidth}
    \centering
    \includegraphics[width=\linewidth]{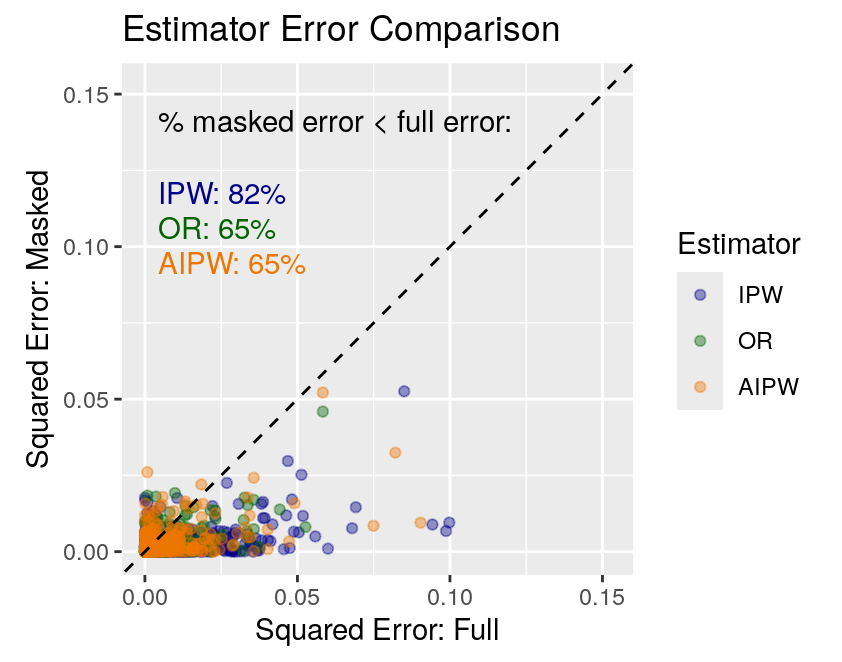}
    \caption{$\beta_5, \sum\alpha = 0.5$}
  \end{subfigure}\hfill
  \begin{subfigure}[b]{0.49\textwidth}
    \centering
    \includegraphics[width=\linewidth]{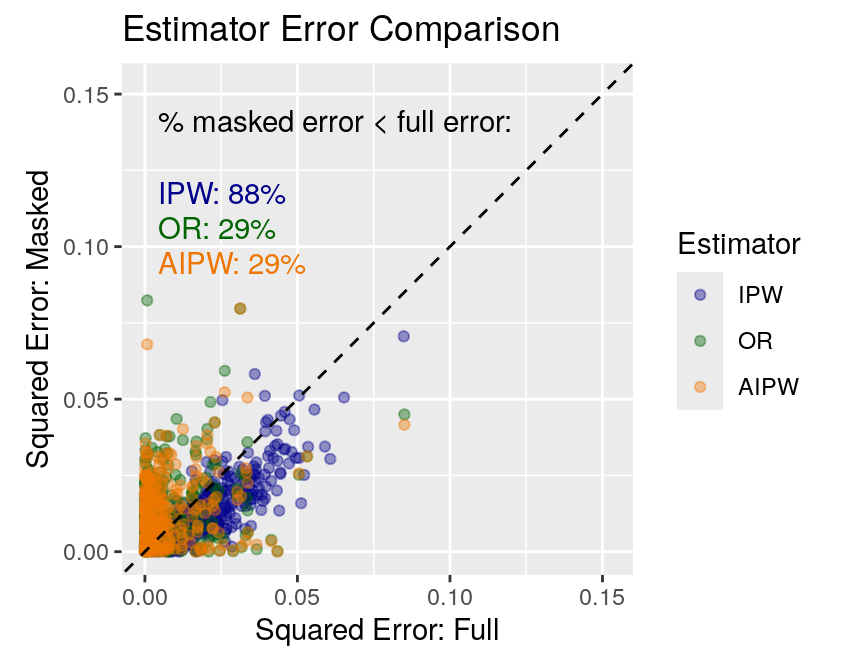}
    \caption{$\beta_5, \sum\alpha = 5$}
  \end{subfigure}

  \vspace{1ex}

  \caption{Error comparison plots for varying $\sum\alpha$ and $\beta$.}
  \label{fig:more}
\end{figure*}

\section{LLM Additional Simulations}

\subsection{Loss Weight Updates}
\label{app:lambda-details}
The weight $\lambda_{Q}$ is a scale-fixing constant. For continuous outcomes, $L_{\text{out}}$ and $L_{\text{CE}}$ can operate on very different scales: $L_{\text{out}}$ is an MSE whose magnitude depends on the variance of $Y$. Without rescaling, the larger term may dominate the gradient and effectively suppress the other. We auto-calibrate $\lambda_{Q}$ before any parameter update by rescaling $$\lambda_{Q} \leftarrow \lambda_{Q} \cdot \bar{L}_{g} / \bar{L}_{Q}$$ where $\bar{L}_{g}$ and $\bar{L}_{Q}$ are the mean per-batch cross entropy and outcome prediction losses at initialization on the training set. After this rescaling, $\lambda_{Q}$ requires no further tuning.
The penalty weight $\lambda_{P}$ is tuned adaptively during training to target a specified validation overlap rate. Let $\mathrm{pct}_{\text{val}}^{(e)}$ denote the fraction of validation data propensities in $[\alpha,\, 1 - \alpha]$ at the end of epoch $e$, and let $q_{e} = 1 - \mathrm{pct}_{\text{val}}^{(e)}$ denote the validation data violation rate. After each epoch we update
\begin{equation}
\label{eq:auto-lambda}
\lambda_{P}^{(e+1)} \;=\; \mathrm{clip}\!\left(\lambda_{P}^{(e)} \cdot \exp\!\big(\eta\,(q_{e} - q^{\star})\big),\;\; \lambda_{\min},\;\; \lambda_{\max}\right),
\end{equation}
with target violation rate $q^\star=0.1$, step size $\eta=1.0$, initialization $\lambda_P^{(0)}=100$, and guardrails $[\lambda_{\min},\lambda_{\max}]=[1,10^5]$.
This multiplicative update increases $\lambda_P$ when too many validation propensities fall outside the overlap band and decreases it when the overlap target is exceeded. We target $90\%$ validation overlap rather than forcing all propensities into $[\alpha,1-\alpha]$ because the overlap penalty is only one part of the causal objective: overly aggressive regularization can suppress treatment-predictive confounding information needed for adjustment\footnote{$90\%$ specifically is chosen as an analogy to commonly used Sturmer trimming, which trims propensity scores outside of the 95th and 5th percentiles of propensity scores, keeping the middle 90\% \citep{sturmer2021propensity}}.

\subsection{LLM Comparators}
\label{app:llm-comp}

We provide further details regarding the implementations of comparison LLM methods.

Two recent methods fine-tune DistilBERT for causal inference with text: TextCause \citep{pryzant-etal-2021-causal} and the TI-estimator \citep{gui2023causalestimationtextdata}. We compare against both because they are the closest existing implementations, and because each makes a different architectural choice that the masking framework speaks to directly. Here, we provide an overview of each method; further details are provided in Appendix~\ref{app:llm-comp}. Default hyperparameters are used.

\paragraph{TextCause} The original paper by \citet{pryzant-etal-2021-causal} introduces a method for learning proxy treatment labels alongside their text adjustment method. However, in the lexicon based treatment setting, we know the labels---it is known which documents have text from which lexicon, and it is reasonable to assume that readers are able to correctly identify this label (ie, $T = D$). Thus, we use their version of the algorithm without proxy label adjustment. Default hyper-parameters (batch size, learning rate, loss head weighting, etc.) are maintained.

TextCause uses an outcome head trained jointly with an MLM auxiliary loss; the loss contains no propensity term (there is code for a head with the ability to predict T, but it contains zero weight in loss function).  

Following \citet{katta2025llm}, we make the following modifications to their code:

\begin{enumerate}
    \item The original code computes the ATE as the difference between the control and treated outcomes; we flip the direction to be the standard treated minus control outcomes in line with most literature.
    \item The original application uses all data for both training and testing. We apply the same data split as in our masking approach, training on the train data and testing on the test data. This provides comparable estimates and avoids inducing post-selection bias.
\end{enumerate}

Furthermore, the original algorithm is implemented for binary outcomes only. To modify the code for continuous outcomes, we replace their cross entropy outcome loss with MSE. Furthermore, the original implementation trains to a suggested default of 3 epochs. Understanding that changing the loss may change the suggested number of epochs, we instead select the number of epochs using a cross validation procedure in line with our masking approach. We train the model on the $80\%$ split, select the number of epochs where the validation loss is minimized, and report the ATE on the test split at that number of epochs.

\paragraph{TI-Estimator} The estimator by \citet{gui2023causalestimationtextdata} is architecturally closer to ours: outcome heads, a jointly trained propensity head that exists to regularize the embedding, and an MLM auxiliary loss. At inference, however, TI-estimator fits a separate propensity model on the two-dimensional outcome summary $\hat{\eta}(W) = (\widehat{Q}_0(W), \widehat{Q}_1(W))$ and plugs it into an AIPW estimator.

As in the TextCause estimator, default hyperparameters are used. Following \citet{katta2025llm}, the following changes to the original code are made in our implementation:

\begin{enumerate}
    \item The propensity score estimators fit the propensity scores on the same data used to estimate treatment effects. Again, we use the same training and test split our masking procedure employs.
    \item We fit propensity scores using the Random Forest Classifier in scitkit-learn.
    \item As discussed in the main body, propensity scores are trimmed and winsorized. Many estimated propensity scores are degenerate (exactly 0 or 1); not doing so leads to unstable and \texttt{NaN} ATE estimates.
\end{enumerate}

\subsection{Amazon simulations: Confounders}
\label{app:sim:confounders}
\begin{longtable}{c >{\raggedright}p{4.5cm} p{7.5cm}}
\caption{Confounding Topic Descriptions} \label{tab:lda_topics} \\
\toprule
\textbf{Topic} & \textbf{Title} & \textbf{Description} \\
\midrule
\endfirsthead

\multicolumn{3}{c}{\tablename\ \thetable\ -- \textit{Continued from previous page}} \\
\toprule
\textbf{Topic} & \textbf{Title} & \textbf{Description} \\
\midrule
\endhead

\midrule
\multicolumn{3}{r}{\textit{Continued on next page}} \\
\endfoot

\bottomrule
\endlastfoot

0  & General Product Reviews       & Overall product satisfaction, purchase experience, shipping, and usage over time \\[3pt]
1  & Supplements \& Energy          & Vitamins, energy supplements, immune support, taste, and ingredient quality \\[3pt]
2  & Dental \& Eyelash Care         & Teeth whitening, toothpaste, flossing, and eyelash curlers \\[3pt]
3  & Household \& Beverages         & Pumps, straws, cups, laundry detergent, coffee, and scented products \\[3pt]
4  & Brushes \& Grooming Tools      & Toothbrushes, pet grooming brushes, trimmers, and cleaning brushes \\[3pt]
5  & Sleep \& Sound Aids            & Earplugs, sound machines, night lights, batteries, and charging accessories \\[3pt]
6  & Mobility \& Compression Wear   & Compression socks, knee braces, walkers, and post-surgery support \\[3pt]
7  & Scales \& Measurement          & Bathroom scales, weight accuracy, digital displays, and readings \\[3pt]
8  & Essential Oils \& Bath          & Essential oils, bath bombs, diffusers, and aromatherapy products \\[3pt]
9  & Water \& Hydration             & Water bottles, hydration reminders, and daily water intake \\[3pt]
10 & Skincare \& Body Wash          & Soaps, lotions, exfoliators, makeup removers, and moisturizers \\[3pt]
11 & Footwear \& Foot Care          & Shoe insoles, heel pads, foot cushions, and plantar fasciitis aids \\[3pt]
12 & Product Failures \& Returns    & Defective products, returns, broken items, and disappointment \\[3pt]
13 & Nail Care Tools                & Nail clippers, files, manicure sets, scissors, and grooming tools \\[3pt]
14 & Massage \& Facial Devices      & Facial and neck massagers, suction tools, and steamers \\[3pt]
15 & Mixed Product Feedback         & General opinions, shipping issues, taste, value, and repeat purchases \\[3pt]
16 & Seating \& Ergonomic Support   & Seat cushions, car accessories, shower chairs, and lumbar support \\[3pt]
17 & Protein \& Meal Powders        & Protein powder, smoothie mixes, flavors, collagen, and meal replacements \\[3pt]
18 & Bags \& Medical Supplies       & Medical bags, supply organizers, nasal/oxygen aids, and bag quality \\[3pt]
19 & Personal Care \& Baby          & Beard care, baby products, shaving cream, shampoo, and skincare \\[3pt]
20 & Gifts \& Packaging             & Gift kits, bags, holiday packaging, party supplies, and value sets \\[3pt]
21 & Masks \& Wearable Accessories  & Face masks, fit and comfort, straps, fabric quality, and sizing \\[3pt]
22 & Specialty Oils \& Wellness     & Tea tree oil, fish oil, yoga mats, wipes, and natural wellness \\[3pt]
23 & Cleaning Products              & Household cleaners, sprays, glass cleaners, car wax, and stain removers \\[3pt]
24 & Eyewear \& Blood Pressure      & Reading glasses, sunglasses, blood pressure monitors, and gloves \\[3pt]
25 & Party Supplies \& Paper Goods  & Birthday decorations, balloons, wrapping paper, tape, and bandages \\[3pt]
26 & Media \& Miscellaneous         & Movies, games, phone accessories, crystals, and language learning \\[3pt]
27 & Fitness \& Pre-Workout         & Pre-workout supplements, creatine, posture correctors, and gym gear \\[3pt]
28 & Replacements \& Razors         & Replacement parts, batteries, razor heads, and brand compatibility \\[3pt]
29 & Home \& Bath Fixtures          & Shower accessories, pillows, stools, air purifiers, and bath fixtures \\[3pt]
30 & Pill Storage \& Sizing         & Pill organizers, container sizes, medication storage, and compartments \\[3pt]
31 & Hair Styling Tools             & Flat irons, curling irons, hair dryers, heat settings, and styling \\[3pt]
32 & Children \& Celebrations       & Baby showers, kids' gifts, wedding supplies, filters, and decor \\[3pt]
33 & Travel \& Vitamin Storage      & Travel cases, gummy vitamins, canes, and portable organizers \\[3pt]
34 & Books \& Stationery            & Books, greeting cards, recipes, photo albums, and paper goods \\
\end{longtable}

\subsection{Compute Resources: LLMs}
\label{compute-llm}

Fine-tuning our masking estimator uses DistilBERT (66M parameters) with batch size 32, max sequence length 128, and up to 20 epochs. A single
fine-tune-plus-evaluation run of the masking and TextCause estimators on the Amazon dataset ($\approx 13,500$ rows) completes in approximately 30 minutes, peaks at around 4~GB of GPU memory, and fits on a single GPU with $\geq 8$~GB of VRAM.

The TI-estimator uses the default hyperparameters provided by the code of \citet{gui2023causalestimationtextdata}. These match our masking configuration except for a batch size of 64. A single fine-tune-plus-evaluation
run of this estimator on the Amazon data completes in approximately 10--30 minutes, peaks at around 9~GB of GPU memory, and fits on a single GPU with $\geq 16$~GB of VRAM.

Reported simulations were run on whichever GPUs were first available on our SLURM cluster: NVIDIA RTX 2080 Ti (11~GB), Tesla P100 (16~GB), RTX A5000 (24~GB), RTX A6000 / RTX 5000 Ada / RTX 6000 Ada (48~GB), and H200 (141~GB). Results are insensitive to GPU choice.

The Amazon simulations comprise 100 iterations $\times$ 4 settings (binary/continuous $\times$ low/high confounding) $\times$ 3 estimators
per setting, totaling approximately 600~GPU-hours. The CFPB application uses 100 bootstrap splits across 3 estimators, totaling approximately
100~GPU-hours. 

Hyperparameter exploration and ablation studies required additional compute hours beyond these reported totals, though the
per-run resource requirements are unchanged.

\subsection{Licenses for LLM Resources}
\label{app:llm-licenses}
 
Table~\ref{tab:llm-licenses} lists the licenses for all external resources used in the LLM experiments (Sections~5 and~6).
Entries marked by \textsuperscript{$\dagger$} are from academic articles with linked GitHub pages who do not identify a license file.
Entry marked by \textsuperscript{$\dagger$} is data publically available from government website. Available at \url{https://www.consumerfinance.gov/data-research/}.
 
\begin{table}[h]
\centering
\caption{Licenses for resources used in LLM experiments.}
\label{tab:llm-licenses}
\begin{tabular}{lll}
\toprule
\textbf{Resource} & \textbf{Reference} & \textbf{License} \\
\midrule
\multicolumn{3}{l}{\textit{Models \& Libraries}} \\
DistilBERT & \citet{sanh2019} & Apache 2.0 \\
HuggingFace Transformers & \citet{wolf2019transformers} & Apache 2.0 \\
PyTorch (incl.\ AdamW) & \citet{paszke2019} & BSD 3-Clause \\
scikit-learn & \citet{scikit-learn} & BSD 3-Clause \\
\midrule
\multicolumn{3}{l}{\textit{Datasets}} \\
Amazon Reviews & \citet{hou2024bridging} & MIT \\
CFPB Complaint Database& -- & \ Government Site\textsuperscript{$\ddagger$} \\
\midrule
\multicolumn{3}{l}{\textit{Packages \& Comparator Code}} \\
\texttt{politeness} R package & \citet{Yeomans2018} & MIT \\
TextCause & \citet{pryzant-etal-2021-causal} & No license identified\textsuperscript{$\dagger$} \\
TI-estimator & \citet{gui2023causalestimationtextdata} & No license identified\textsuperscript{$\dagger$} \\
\bottomrule
\end{tabular}
\end{table}

\subsection{Ablation: Masking Components }
\label{app:ab-mask-comp}

The LLM masking mechanism combines two components: (1) Replacement masking by treatment-defining tokens during training and (2) the adaptive overlap penalty of  $L_{\textrm{penalty}}$.  This section studies the contribution of each component to model performance.

\subsubsection{Masking Alone; No Penalty}
\label{app:no-pen}

Tables \ref{tab:bin-no-p} and \ref{tab:cont-no-p} compare the full method (Masking Auto) against an ablation that trains on masked text but does not include the penalty loss (Masking: No Penalty). We find that token masking alone reduces bias and MSE relative to the naive estimator, TextCause, and the TI variants across three out of four settings, \textit{\textbf{confirming that removing the lexical shortcut during training is a primary driver of improvement}}. Compared to masking with a penalty, no penalty masking has in different settings similar, lower, and higher MSE, though it is generally close.

The overlap diagnostic require careful interpretation. The median fraction of propensity scores in $[0.1, 0.9]$ is 16-27\% without the penalty, compared to 87-98\% with it. However, $[0.1, 0.9]$ is a diagnostic threshold, not the true overlap band $[\alpha^*, 1 - \alpha^*]$, which is unknown. Propensity scores outside $[0.1, 0.9]$ are not necessarily overlap violations if the true confounding structure supports more extreme values. The more informative comparison is between masking without penalty and the unmasked TI-estimator. Figure \ref{fig:prop-grid} shows the contrast clearly: the TI-estimator produces propensity distributions that are effectively binary, with tall spikes at 0 and 1 and little mass between, while the masking-only estimator yields smoother, overlapping distributions across the unit interval. In fact,  in each of the settings, at least 85\% of TI propensity scores are exactly 0 or 1, reflecting deterministic treatment encoding. By contrast, masking without penalty produces zero degenerate propensity scores. The distributions retain spread and maintain (imperfect) separation between treatment arms, but without the point masses at the boundaries that indicate complete overlap collapse.

The penalty provides additional control for practitioners who require propensity scores within a specific band, such as the $[0.1, 0.9]$ range commonly used for trimming or winsorization, but it is not necessary to avoid the degenerate encoding that motivates this paper.

\begin{table}[h]\centering\small
\begin{tabular}{@{}l cccc cccc@{}}
\toprule
& \multicolumn{4}{c}{\textbf{Low Confounding}} & \multicolumn{4}{c}{\textbf{High Confounding}} \\
\cmidrule(lr){2-5}\cmidrule(lr){6-9}
Estimator & MSE & Bias & Var & PS & MSE & Bias & Var & PS \\
\midrule
True & .0000 & .0001 & .0000 & -- & .0000 & .0006 & .0000 & -- \\
Naive & .0230 & .1514 & .0001 & -- & .0553 & .2349 & .0001 & -- \\
\midrule
Masking (Auto) & .0038 & .0287 & .0030 & 96.0\% & .0115 & .1000 & .0015 & 97.7\% \\
Masking: No Penalty & .0037 & .0362 & .0024 & 26.4\% & .0051 & .0524 & .0024 & 26.5\% \\
TextCause & .0237 & .1531 & .0002 & 100.0\% & .0565 & .2371 & .0003 & 100.0\% \\
TI (Trimmed) & .0239 & .0509 & .0213 & 2.2\% & .0280 & .0491 & .0256 & 2.2\% \\
TI (Winsorized) & .0363 & .1693 & .0076 & 2.2\% & .0436 & .1908 & .0072 & 2.2\% \\
TI (Out. Reg.) & .0434 & .1887 & .0078 & 2.2\% & .0523 & .2124 & .0072 & 2.2\% \\
\bottomrule
\end{tabular}

\caption{Binary outcomes: comparison of estimators, including an iteration with no penalty loss.}
\label{tab:bin-no-p}
\end{table}

\begin{table}[h]\centering\small
\begin{tabular}{@{}l cccc cccc@{}}
\toprule
& \multicolumn{4}{c}{\textbf{Low Confounding}} & \multicolumn{4}{c}{\textbf{High Confounding}} \\
\cmidrule(lr){2-5}\cmidrule(lr){6-9}
Estimator & MSE & Bias & Var & PS & MSE & Bias & Var & PS \\
\midrule
True & .0229 & .0186 & .0226 & -- & .0229 & .0186 & .0226 & -- \\
Naive & .3669 & .6058 & .0000 & -- & 1.1238 & 1.0601 & .0000 & -- \\
\midrule
Masking (Auto) & .0195 & -.0425 & .0177 & 87.0\% & .0442 & -.0674 & .0397 & 87.9\% \\
Masking: No Penalty & .0136 & .0330 & .0125 & 16.4\% & .0515 & .1603 & .0258 & 17.2\% \\
TextCause & .2746 & .5219 & .0021 & 100.0\% & .6132 & .7779 & .0081 & 100.0\% \\
TI (Trimmed) & .0336 & -.1072 & .0221 & 2.9\% & .0615 & -.1614 & .0355 & 5.5\% \\
TI (Winsorized) & .0248 & .1420 & .0047 & 2.9\% & .0422 & .1626 & .0157 & 5.5\% \\
TI (Out. Reg.) & .0380 & .1820 & .0049 & 2.9\% & .0497 & .1789 & .0177 & 5.5\% \\
\bottomrule
\end{tabular}

\caption{Continuous outcomes: comparison of estimators, including an iteration with no penalty loss.}
\label{tab:cont-no-p}
\end{table}

\begin{figure}[htbp]
    \centering
    \begin{tabular}{cc}
        \includegraphics[width=0.47\textwidth]{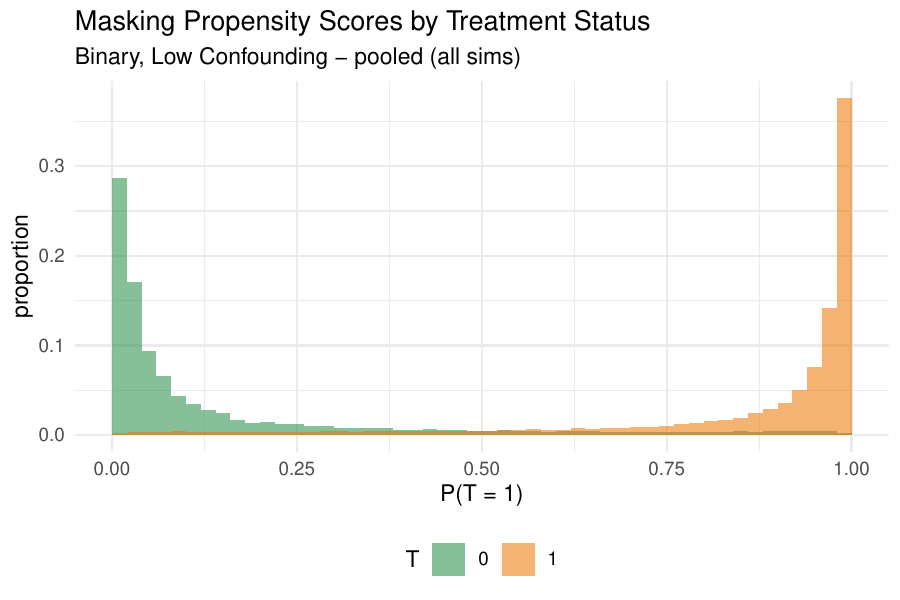} &
        \includegraphics[width=0.47\textwidth]{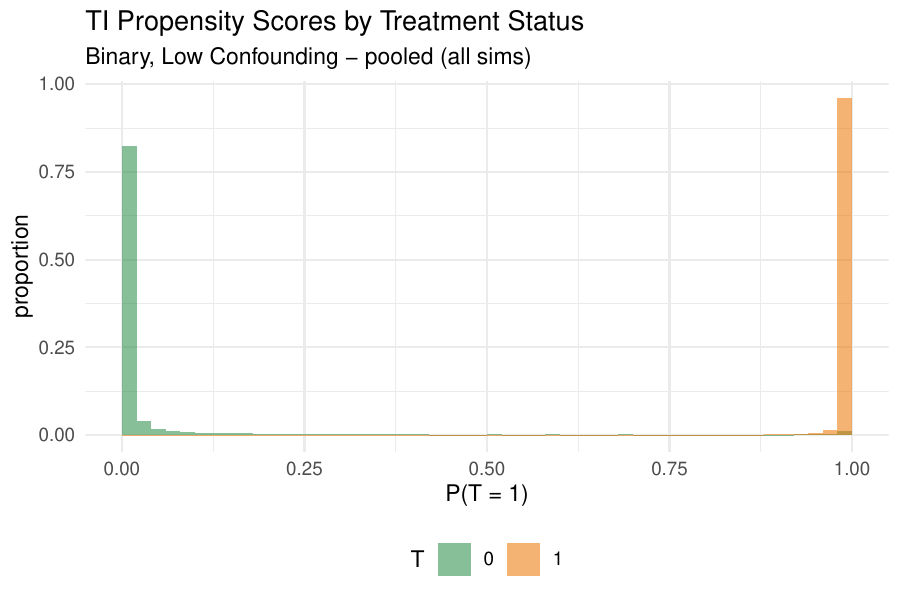} \\
        \includegraphics[width=0.47\textwidth]{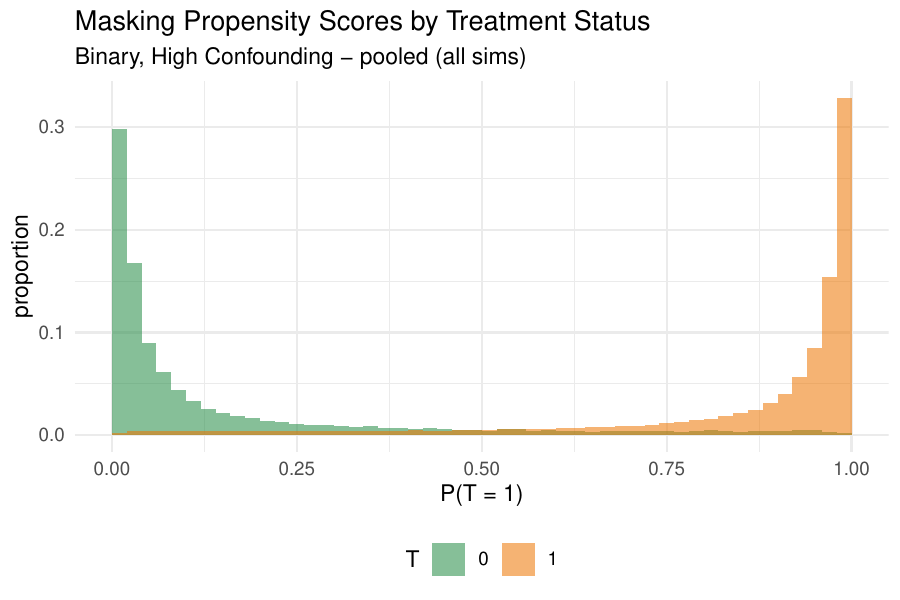} &
        \includegraphics[width=0.47\textwidth]{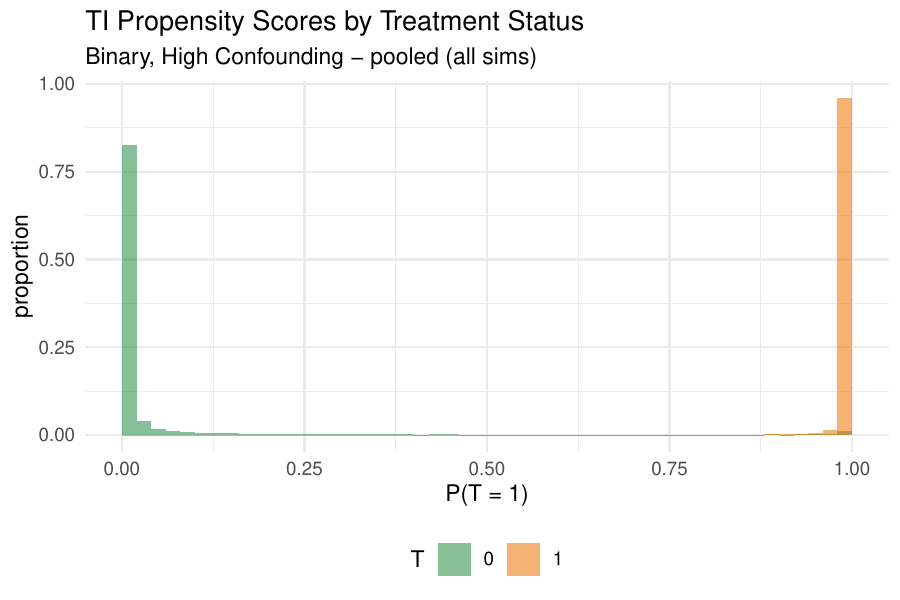} \\
        \includegraphics[width=0.47\textwidth]{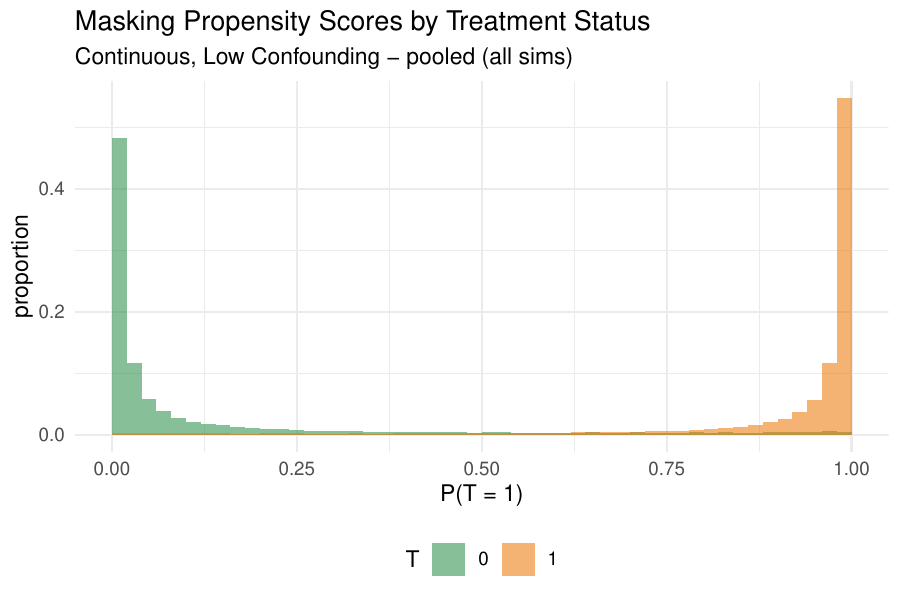} &
        \includegraphics[width=0.47\textwidth]{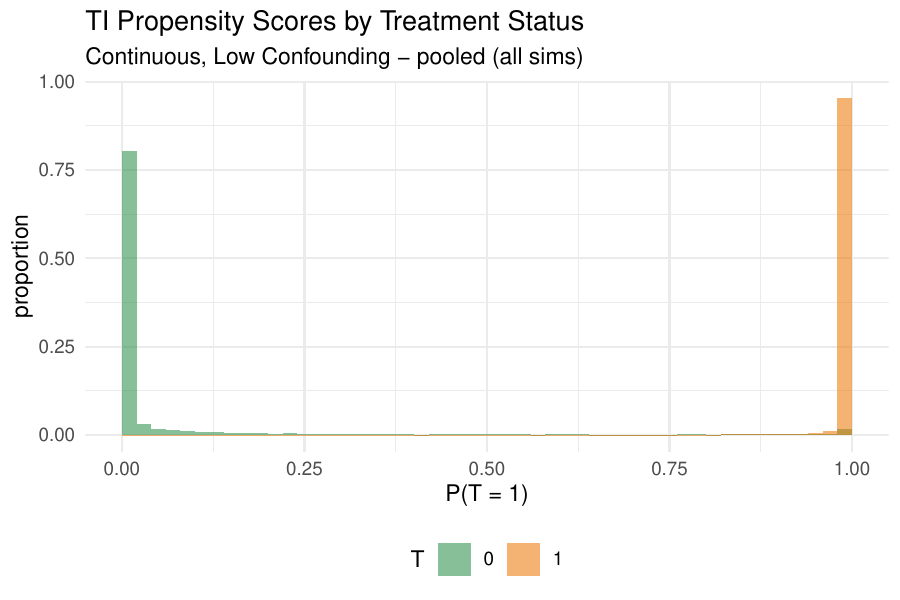} \\
        \includegraphics[width=0.47\textwidth]{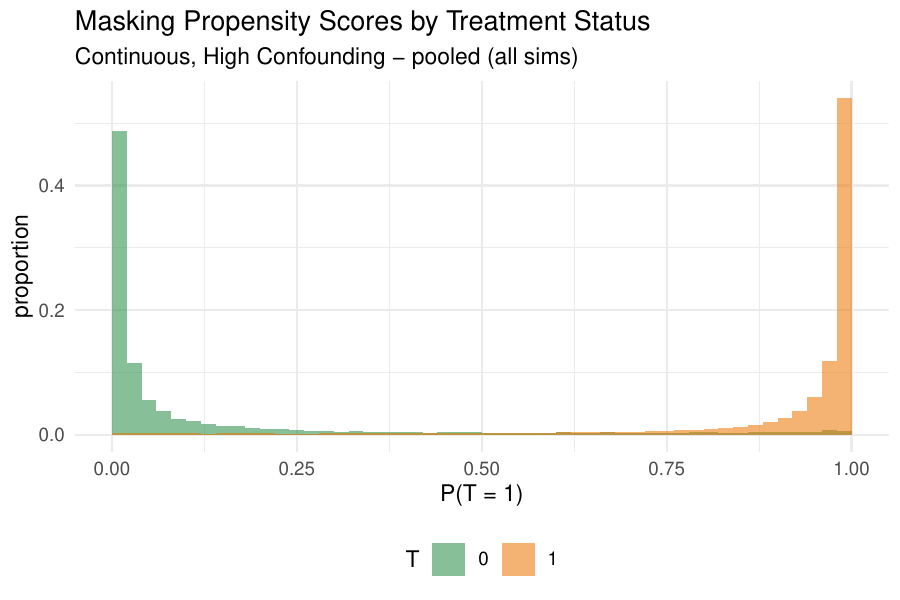} &
        \includegraphics[width=0.47\textwidth]{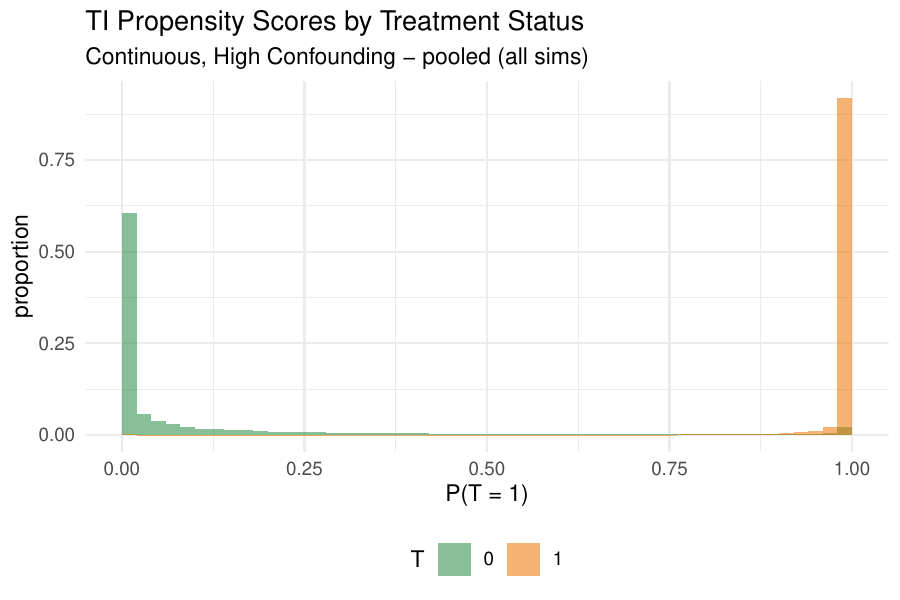} \\
    \end{tabular}
    \caption{Estimated propensity score distributions for the masking-only estimator (no penalty; left column) and the TI-estimator (right column), pooled across all 100 simulation iterations. Rows correspond to binary low confounding, binary high confounding, continuous low confounding, and continuous high confounding. Colors indicate treatment status. The TI-estimator concentrates propensity mass at 0 and 1, reflecting deterministic treatment encoding from unmasked text. The masking-only estimator produces no degenerate scores and maintains overlapping support across the unit interval. Notice differences in $y$-axis scales.}
    \label{fig:prop-grid}
\end{figure}

\subsubsection{Penalty Alone; No Masking}

\begin{table}[h]\centering\small
\begin{tabular}{@{}l cccc cccc@{}}
\toprule
& \multicolumn{4}{c}{\textbf{Low Confounding}} & \multicolumn{4}{c}{\textbf{High Confounding}} \\
\cmidrule(lr){2-5}\cmidrule(lr){6-9}
Estimator & MSE & Bias & Var & PS & MSE & Bias & Var & PS \\
\midrule
True & .0000 & .0001 & .0000 & -- & .0000 & .0006 & .0000 & -- \\
Naive & .0230 & .1514 & .0001 & -- & .0553 & .2349 & .0001 & -- \\
\midrule
Masking (Auto) & .0038 & .0287 & .0030 & 96.0\% & .0115 & .1000 & .0015 & 97.7\% \\
Masking - Penalty Only & .0110 & .0980 & .0014 & 76.4\% & .0221 & .1435 & .0015 & 79.2\% \\
TextCause & .0237 & .1531 & .0002 & 100.0\% & .0565 & .2371 & .0003 & 100.0\% \\
TI (Trimmed) & .0239 & .0509 & .0213 & 2.2\% & .0280 & .0491 & .0256 & 2.2\% \\
TI (Winsorized) & .0363 & .1693 & .0076 & 2.2\% & .0436 & .1908 & .0072 & 2.2\% \\
TI (Out. Reg.) & .0434 & .1887 & .0078 & 2.2\% & .0523 & .2124 & .0072 & 2.2\% \\
\bottomrule
\end{tabular}
\caption{Binary outcomes: comparison of estimators, including with penalty and no masking.}
\label{tab:bin-pen-only}
\end{table}

\begin{table}[h]\centering\small
\begin{tabular}{@{}l cccc cccc@{}}
\toprule
& \multicolumn{4}{c}{\textbf{Low Confounding}} & \multicolumn{4}{c}{\textbf{High Confounding}} \\
\cmidrule(lr){2-5}\cmidrule(lr){6-9}
Estimator & MSE & Bias & Var & PS & MSE & Bias & Var & PS \\
\midrule
True & .0229 & .0186 & .0226 & -- & .0229 & .0186 & .0226 & -- \\
Naive & .3669 & .6058 & .0000 & -- & 1.1238 & 1.0601 & .0000 & -- \\
\midrule
Masking (Auto) & .0195 & -.0425 & .0177 & 87.0\% & .0442 & -.0674 & .0397 & 87.9\% \\
Masking - Penalty Only & .0158 & -.0371 & .0144 & 90.1\% & .0651 & .1519 & .0420 & 90.7\% \\
TextCause & .2746 & .5219 & .0021 & 100.0\% & .6132 & .7779 & .0081 & 100.0\% \\
TI (Trimmed) & .0336 & -.1072 & .0221 & 2.9\% & .0615 & -.1614 & .0355 & 5.5\% \\
TI (Winsorized) & .0248 & .1420 & .0047 & 2.9\% & .0422 & .1626 & .0157 & 5.5\% \\
TI (Out. Reg.) & .0380 & .1820 & .0049 & 2.9\% & .0497 & .1789 & .0177 & 5.5\% \\
\bottomrule
\end{tabular}
\caption{Continuous outcomes: comparison of estimators, including with penalty and no masking.}
\label{tab:cont-pen-only}
\end{table}

Tables ~\ref{tab:bin-pen-only} and \ref{tab:cont-pen-only} report an ablation that trains on unmasked text but retains the adaptive overlap penalty (Masking: Penalty Only). Unlike the full masking pipeline, this variant requires no treatment-defining lexicon: the penalty regularizes propensity scores toward the interior of $[0,1]$ regardless of the source of treatment encoding.  This makes it applicable, in principle, to any text-as-treatment setting, including those where treatment is not lexically defined.

The penalty alone improves overlap relative to the TI-estimator, with 76--90\% of propensity scores in $[0.1, 0.9]$. However, because the encoder still receives treatment-defining tokens as input, the penalty must work against a strong gradient signal pushing the representation toward encoding treatment status. MSE is reduced relative to the naive estimator, TextCause, and most TI variants in all settings. In three of four settings, the full masking pipeline achieves lower MSE, confirming that removing the lexical shortcut at the input level and regularizing the propensity score are complementary.  In one setting (continuous, low confounding), the penalty-only variant achieves MSE of .0158 compared to .0195 for masking, suggesting that the penalty can occasionally provide sufficient regularization on its own.

The adaptive update rule provides further insight into the division of labor between masking and the penalty.  When training on masked text, $\lambda_P$ decreases from its initial value of 100 in a majority of iterations (median final value: 78; minimum: 26), indicating that with the lexical shortcut removed, the encoder produces moderate propensity scores without strong regularization.  When training on unmasked text (penalty only), $\lambda_P$ increases in every iteration (minimum final value: 165; maximum: 738), reflecting the substantially stronger regularization needed to counteract the gradient signal from treatment-defining tokens.  This asymmetry confirms that the penalty works harder in the absence of masking. It must suppress treatment encoding that masking would have prevented at the input level.

These results frame the penalty not merely as a component of the masking pipeline, but as a lexicon-free tool for mitigating overlap collapse in text-as-treatment problems.  While it lacks the theoretical guarantees available for lexicon-based masking, it offers a practical method when no lexicon is available.

To investigate the mechanism responsible for this, we perform an exploratory analysis of how the encoder treats $\mathcal{L}$ versus \texttt{[MASK]} tokens in producing the \texttt{[CLS]} representation. For a fixed held-out set of test documents, we tokenize each twice — once as $\mathbf{W}_i$ and once as $\mathbf{W}_i^{\mathrm{mask}}$ — and at every epoch compute the cosine similarity between the average \texttt{[CLS]} vectors produced from the two tokenizations.

A value near 1 indicates the encoder has learned an invariance to the $\mathcal{L} \to \texttt{[MASK]}$ swap; a value near 0 indicates that masking lexicon tokens substantially shifts the representation. We compare three training regimes: train on $\mathbf{W}^{\mathrm{mask}}$ with $L_\text{penalty}$ (production), train on $\mathbf{W}$ with $L_\text{penalty}$ (penalty-only ablation), and train on $\mathbf{W}$ without $L_\text{penalty}$.

\begin{wrapfigure}{r}{0.5\linewidth}
    \centering
    \includegraphics[width=0.5\textwidth]{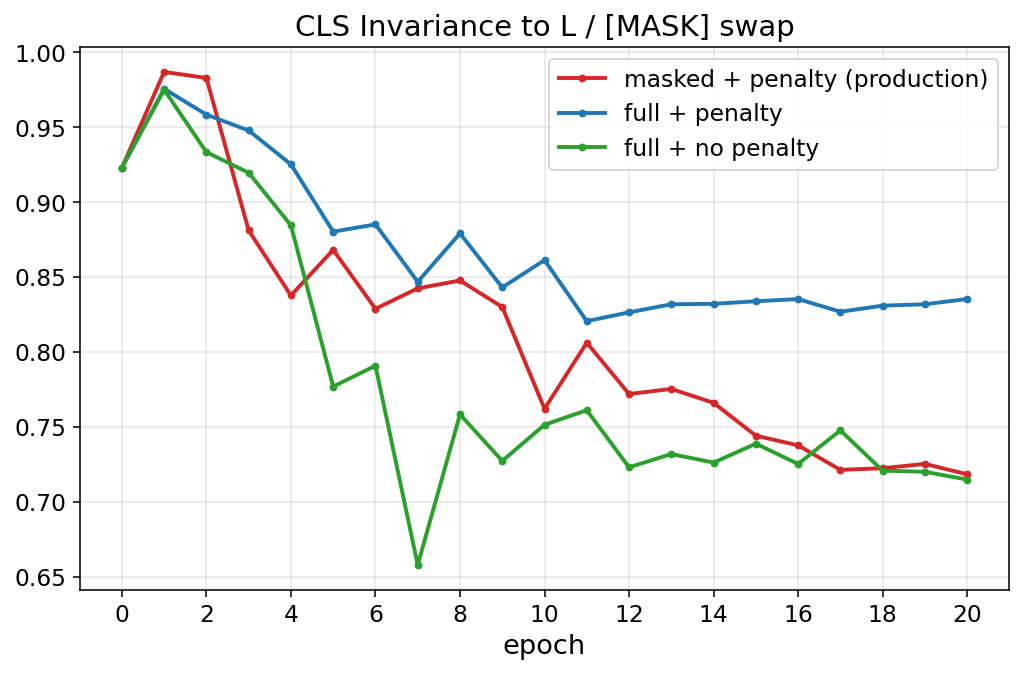}
    \caption{Exploratory; single simulation on binary high confounding data.}
    \label{fig:cls-drift}
\end{wrapfigure}

\textit{Interpretation.} Figure \ref{fig:cls-drift} shows exploratory results from a single simulation on the binary, high-confounding Amazon dataset; replicating across simulations and confounding settings is a direction for future work.

Training on $\mathbf{W}$ with $L_\text{penalty}$ sustains a high cosine across epochs ($\sim 0.84$), indicating that the penalty induces the encoder to treat $\mathcal{L}$ and \texttt{[MASK]} similarly at the \texttt{[CLS]}.

Training on $\mathbf{W}$ without $L_\text{penalty}$ shows the largest representational gap, consistent with the lexicon carrying information whose removal substantively shifts the representation.

The production regime often falls between the two which we read as further support for the train-on-masked, evaluate-on-full design: evaluation on $\mathbf{W}$ requires the encoder to set aside its learned representation of the mask pattern and re-engage pretrained knowledge of lexicon tokens. 

We treat these readings as hypotheses pending further study, and note that they suggest a representation-level analysis of $L_\text{penalty}$ as a potential path toward extending the framework beyond lexicon-based treatments.

\paragraph{Summary.}
The two components address complementary failure modes. Masking removes the direct lexical path from treatment-defining tokens to the representation, eliminating the dominant source of overlap collapse. The penalty regularizes residual treatment signal that survives masking, ensuring propensity scores remain in a range suitable for stable inference. Their combination is a reliable method for practitioners to encode textual confounders. We suggest a two-tier recommendation:
\begin{enumerate}
    \item \textbf{Lexicon available:} Apply masking and the overlap penalty jointly for the strongest and most theoretically grounded performance.
    \item \textbf{No lexicon available:} The overlap penalty alone provides a lexicon-free alternative that still substantially reduces bias and improves overlap relative to existing text-as-treatment methods.
\end{enumerate}

\subsection{Evaluation on Masked Text}
\label{app:eval-mask}

Section \ref{sec:masking} introduced the design choice to train on masked text but
evaluate on the original  document. Evaluating on
$\mathbf{W}$ may be beneficial as the mask-trained encoder has not
learned to exploit lexicon tokens, yet can still draw on non-lexicon
context that those tokens provide.

Table \ref{tab:results_mask_eval} compares the results of the proposed LLM procedure from Section~\ref{sec:llm} evaluating on masked text compared to full text. In three out of four cases, we see a lower MSE when evaluating on full text. This is empirical support that the mask-trained encoder, having encountered \texttt{[MASK]} tokens throughout training, is sensitive to their presence and frequency. Full-text evaluation achieves lower MSE in
three of four settings, with the largest gains under continuous outcomes: MSE drops from 0.0637 to 0.0195 under low confounding and from 0.1211 to 0.0442 under high confounding. The overlap diagnostics shift correspondingly, with the fraction of propensity scores in
$[0.1, 0.9]$ rising from 17--18\% under masked evaluation to 87--88\% under full-text evaluation. Under binary high confounding, masked evaluation achieves a lower MSE (0.0053 vs.\ 0.0115), though both values represent substantial improvement over the naive baseline.
The continuous settings, where the outcome model is more sensitive to representation quality, show the clearer separation between the two strategies and provide the stronger evidence for full-text evaluation as the default.

\begin{table}[h]\centering\small
\begin{tabular}{@{}l cccc cccc@{}}
\toprule
& \multicolumn{4}{c}{\textbf{Low Confounding}} & \multicolumn{4}{c}{\textbf{High Confounding}} \\
\cmidrule(lr){2-5}\cmidrule(lr){6-9}
Estimator & MSE & Bias & Var & PS & MSE & Bias & Var & PS \\
\midrule
\rowcolor{gray!15}
\multicolumn{9}{@{}l}{\textit{Binary Outcomes}} \\[.3em]
True & .0000 & .0001 & .0000 & -- & .0000 & .0006 & .0000 & -- \\
Naive & .0230 & .1514 & .0001 & -- & .0553 & .2349 & .0001 & -- \\
\cmidrule(l){1-9}
Masking (Full Eval) & .0038 & .0287 & .0030 & 96.0\% & .0115 & .1000 & .0015 & 97.7\% \\
Masking (Mask Eval) & .0048 & .0397 & .0032 & 99.3\% & .0053 & .0515 & .0027 & 99.6\% \\

\midrule
\rowcolor{gray!15}
\multicolumn{9}{@{}l}{\textit{Continuous Outcomes}} \\[.3em]
True & .0229 & .0186 & .0226 & -- & .0229 & .0186 & .0226 & -- \\
Naive & .3669 & .6058 & .0000 & -- & 1.1238 & 1.0601 & .0000 & -- \\
\cmidrule(l){1-9}
Masking (Full Eval) & .0195 & -.0425 & .0177 & 87.0\% & .0442 & -.0674 & .0397 & 87.9\% \\
Masking adapt (Mask Eval) & .0637 & .2409 & .0057 & 17.0\% & .1211 & .3163 & .0211 & 18.2\% \\
\bottomrule
\end{tabular}
\caption{Binary (top) and continuous (bottom) outcomes using the full proposed loss function.  Full Eval rows are evaluated on full text; Mask Eval rows are evaluate on masked text. \label{tab:results_mask_eval}}
\end{table}

\subsection{Sensitivity to $\lambda_p$}

The overlap penalty weight $\lambda_p$ is the primary hyperparameter introduced by our method. Appendix~\ref{app:no-pen} established that replacement masking alone can be beneficial for ATE estimation; the penalty provides additional control over the propensity score distribution. Here, we investigate how the choice of $\lambda_p$ affects this control and
whether performance is sensitive to its value. We examine the adaptive update rule proposed as our main method with varying initial values $\lambda_p^{(0)}$. As further comparison, we provide the same sweep when evaluating on masked data.

\subsubsection{Adaptive Starting Values}
\label{app:auto-lambda-p-sens}
Tables~\ref{tab:sweep-auto-bin} and~\ref{tab:sweep-auto-cont} report results across $\lambda_p^{(0)} \in \{10, 30, 100, 300, 1000\}$ using the adaptive update rule. Two patterns emerge. First, the fraction of propensity scores in $[0.1, 0.9]$ increases monotonically with $\lambda_p^{(0)}$: stronger initial penalization drives propensity scores toward the interior of the unit interval more quickly, and lower starting values may not ramp up fast enough before validation loss is minimized and training stops. Second, MSE does not decrease
monotonically. Performance improves as $\lambda_p^{(0)}$ increases from 10 to an intermediate range (100--300), then degrades at 1000. This reflects a tradeoff: \textit{too little penalization allows residual treatment leakage, while too much suppresses treatment-predictive confounding information that the representation needs for outcome adjustment}.

Figure~\ref{fig:ps-mse-line-tradeoff} summarizes this tradeoff by plotting MSE against the median overlap rate for each starting value and setting. The relationship traces an approximate U-shape in MSE as the overlap rate increases, with the minimum consistently falling in the 87--99\% range. The pattern is qualitatively similar across binary and
continuous outcomes and across confounding strengths, though there appears to be some noise and sharp spikes in the binary high confounding setting, suggesting that moderate penalization enough to prevent extreme propensities but is robust across settings. 
The default starting value of $\lambda_p^{(0)} = 100$ falls within or near this favorable range in all four settings. While binary high confounding actually experiences a spike in MSE at $\lambda_p^{(0)} = 100$, we still see that this method strongly outperforms competitors. The total range of MSE across all starting values in this setting is 0.0070 to 0.0141, a
window small enough that simulation noise can shift the ranking. Every value of $\lambda_p^{(0)}$ in this range reduces MSE by at least half relative to the strongest competitor, and all achieve overlap rates above 78\%. The default starting value of $\lambda_p^{(0)} = 100$ falls within or near the favorable range in all four settings, and where it is not the single best value, the differences are small relative to the gains over competing methods.

\begin{table}[h]\centering\small
\begin{tabular}{@{}l cccc cccc@{}}
\toprule
& \multicolumn{4}{c}{\textbf{Low Confounding}} & \multicolumn{4}{c}{\textbf{High Confounding}} \\
\cmidrule(lr){2-5}\cmidrule(lr){6-9}
Estimator & MSE & Bias & Var & PS & MSE & Bias & Var & PS \\
\midrule
Masking ($\lambda_p^{(0)} = 10$) & .0047 & .0517 & .0020 & 69.6\% & .0076 & .0737 & .0022 & 78.3\% \\
Masking ($\lambda_p^{(0)} = 30$) & .0051 & .0573 & .0018 & 86.9\% & .0089 & .0839 & .0018 & 89.7\% \\
Masking ($\lambda_p^{(0)} = 100$) & .0038 & .0287 & .0030 & 96.0\% & .0115 & .1000 & .0015 & 97.7\% \\
Masking ($\lambda_p^{(0)} = 300$) & .0036 & .0093 & .0035 & 99.9\% & .0070 & .0681 & .0023 & 99.9\% \\
Masking ($\lambda_p^{(0)} = 1000$) & .0054 & .0445 & .0034 & 100.0\% & .0141 & .1067 & .0027 & 100.0\% \\
\bottomrule
\end{tabular}
\caption{Binary outcomes: masking $\lambda_p^{(0)}$ sweep.}
\label{tab:sweep-auto-bin}
\end{table}

\begin{table}[h]\centering\small
\begin{tabular}{@{}l cccc cccc@{}}
\toprule
& \multicolumn{4}{c}{\textbf{Low Confounding}} & \multicolumn{4}{c}{\textbf{High Confounding}} \\
\cmidrule(lr){2-5}\cmidrule(lr){6-9}
Estimator & MSE & Bias & Var & PS & MSE & Bias & Var & PS \\
\midrule
Masking ($\lambda_p^{(0)} = 10$) & .0731 & -.2020 & .0323 & 86.5\% & .1444 & -.3162 & .0444 & 84.4\% \\
Masking ($\lambda_p^{(0)} = 30$) & .0484 & -.1492 & .0261 & 86.8\% & .0732 & -.1834 & .0396 & 83.1\% \\
Masking ($\lambda_p^{(0)} = 100$) & .0195 & -.0425 & .0177 & 87.0\% & .0442 & -.0674 & .0397 & 87.9\% \\
Masking ($\lambda_p^{(0)} = 300$) & .0152 & -.0261 & .0145 & 91.1\% & .0488 & .1198 & .0344 & 97.8\% \\
Masking ($\lambda_p^{(0)} = 1000$) & .0176 & .0439 & .0157 & 98.4\% & .0513 & .0516 & .0486 & 99.9\% \\
\bottomrule
\end{tabular}
\caption{Continuous outcomes: masking $\lambda_p^{(0)}$ sweep.}
\label{tab:sweep-auto-cont}
\end{table}

\begin{figure}
    \centering
    \includegraphics[width=0.95\linewidth]{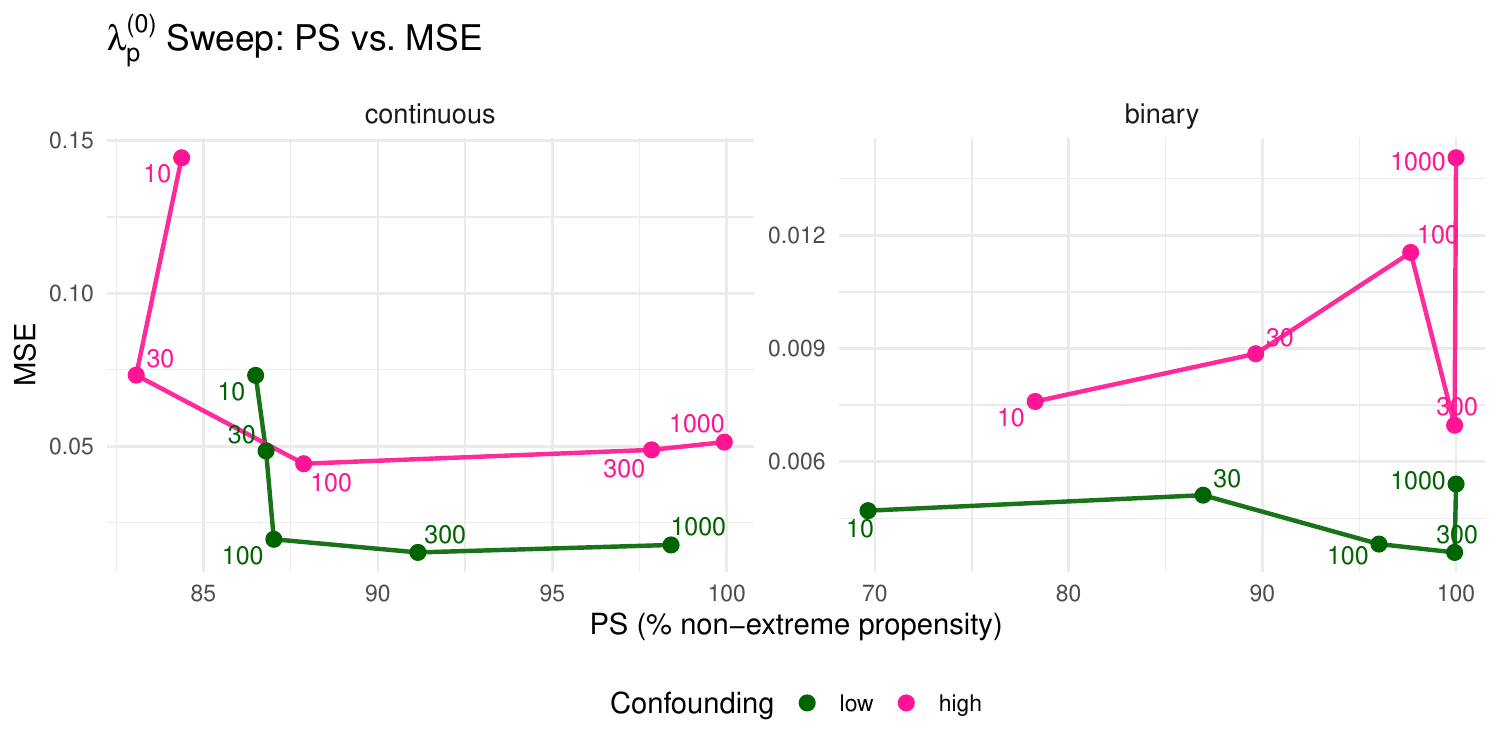}
    \caption{MSE versus median overlap rate (fraction of propensity scores in $[0.1, 0.9]$) across adaptive starting values $\lambda_p^{(0)} \in \{10, 30, 100, 300, 1000\}$, evaluated on full text.}

    \label{fig:ps-mse-line-tradeoff}
\end{figure}

\subsubsection{Adaptive Starting Values: Masked Evaluation}

Tables \ref{tab:mask-eval-bin-sweep} and \ref{tab:mask-eval-cont-sweep} repeat the $\lambda_p^{(0)}$ sweep from above but evaluate the trained encoder on masked text rather than full text. The comparison indicates how the evaluation strategy interacts with penalty strength. Two differences from the full-text evaluation sweep stand out. First, MSEs are generally higher across starting values and settings, consistent with that masked evaluation can expose the mask pattern as a treatment signal. Second, the overlap diagnostic responds more sluggishly to the penalty: even at $\lambda_p^{(0)}$ = 1000, the median fraction of propensity scores in [0.1, 0.9] reaches only 45–50\% for continuous outcomes.


\begin{table}[h]\centering\small
\begin{tabular}{@{}l cccc cccc@{}}
\toprule
& \multicolumn{4}{c}{\textbf{Low Confounding}} & \multicolumn{4}{c}{\textbf{High Confounding}} \\
\cmidrule(lr){2-5}\cmidrule(lr){6-9}
Estimator & MSE & Bias & Var & PS & MSE & Bias & Var & PS \\
\midrule
Masking ($\lambda_p^{(0)} = 10$) & .0119 & .0996 & .0019 & 17.7\% & .0219 & .1418 & .0018 & 17.6\% \\
Masking ($\lambda_p^{(0)} = 30$) & .0082 & .0748 & .0026 & 20.7\% & .0248 & .1518 & .0018 & 25.6\% \\
Masking ($\lambda_p^{(0)} = 100$) & .0090 & .0829 & .0021 & 42.2\% & .0172 & .1214 & .0025 & 48.3\% \\
Masking ($\lambda_p^{(0)} = 300$) & .0107 & .0945 & .0018 & 97.1\% & .0215 & .1393 & .0021 & 95.2\% \\
Masking ($\lambda_p^{(0)} = 1000$) & .0077 & .0586 & .0043 & 99.7\% & .0187 & .1250 & .0031 & 99.9\% \\
\bottomrule
\end{tabular}
\caption{Binary outcomes evaluated on masked text: masking $\lambda_p^{(0)}$ sweep}
\label{tab:mask-eval-bin-sweep}
\end{table}

\begin{table}[h]\centering\small
\begin{tabular}{@{}l cccc cccc@{}}
\toprule
& \multicolumn{4}{c}{\textbf{Low Confounding}} & \multicolumn{4}{c}{\textbf{High Confounding}} \\
\cmidrule(lr){2-5}\cmidrule(lr){6-9}
Estimator & MSE & Bias & Var & PS & MSE & Bias & Var & PS \\
\midrule
Masking ($\lambda_p^{(0)} = 10$) & .1473 & .3790 & .0036 & 13.4\% & .3251 & .5608 & .0106 & 12.9\% \\
Masking ($\lambda_p^{(0)} = 30$) & .0515 & .2129 & .0062 & 13.9\% & .2467 & .4842 & .0123 & 15.0\% \\
Masking ($\lambda_p^{(0)} = 100$) & .0637 & .2409 & .0057 & 17.0\% & .1211 & .3163 & .0211 & 18.2\% \\
Masking ($\lambda_p^{(0)} = 300$) & .0476 & .2021 & .0068 & 22.6\% & .1388 & .3437 & .0207 & 44.1\% \\
Masking ($\lambda_p^{(0)} = 1000$) & .0366 & .1716 & .0072 & 45.1\% & .0497 & .1244 & .0342 & 46.8\% \\
\bottomrule
\end{tabular}
\caption{Continuous outcomes evaluated on masked text: masking $\lambda_p^{(0)}$ sweep.}
\label{tab:mask-eval-cont-sweep}
\end{table}

\subsubsection{Fixed Starting Values}

What if instead of an automatically updating $\lambda_p$, we used fixed values?

Tables~\ref{tab:fixed-bin} and~\ref{tab:fixed-cont} report results for fixed $\lambda_P \in \{10, 30, 100, 300, 1000\}$. Fixed values in the range $10 - 300$ consistently outperform comparators from Table~\ref{tab:results_adaptive_lambda} across both outcome types and confounding strengths, confirming that the method is not sensitive to precise tuning. With excessive penalization ($\lambda_P = 1000$), overlap reaches 100\% but MSE increases, particularly under continuous high confounding.  This is consistent with the bias-overlap tradeoff formalized in Section~\ref{sec:theory}. Too little regularization permits residual treatment leakage, while too much suppresses the treatment-predictive confounding signal needed for outcome adjustment.

This performance suggest practitioners can choose a fixed value without elaborate tuning.  The adaptive rule provides a convenience by automating selection within this range but is not essential to the method's performance.

\begin{table}[h]\centering\small
\begin{tabular}{@{}l cccc cccc@{}}
\toprule
& \multicolumn{4}{c}{\textbf{Low Confounding}} & \multicolumn{4}{c}{\textbf{High Confounding}} \\
\cmidrule(lr){2-5}\cmidrule(lr){6-9}
Estimator & MSE & Bias & Var & PS & MSE & Bias & Var & PS \\
\midrule
Masking ($\lambda_p$=10)   & .0054 & .0565 & .0022 & 66.1\% & .0074 & .0729 & .0021 & 73.8\% \\
Masking ($\lambda_p$=30)   & .0059 & .0632 & .0019 & 90.8\% & .0094 & .0869 & .0018 & 92.3\% \\
Masking ($\lambda_p$=100)  & .0048 & .0397 & .0032 & 99.3\% & .0053 & .0515 & .0027 & 99.6\% \\
Masking ($\lambda_p$=300)  & .0037 & .0158 & .0034 & 100.0\% & .0082 & .0767 & .0023 & 100.0\% \\
Masking ($\lambda_p$=1000) & .0064 & .0553 & .0033 & 100.0\% & .0175 & .1212 & .0028 & 100.0\% \\
\bottomrule
\end{tabular}
\caption{Binary outcomes: fixed $\lambda_p$ sweep.}
\label{tab:fixed-bin}
\end{table}

\begin{table}[h]\centering\small
\begin{tabular}{@{}l cccc cccc@{}}
\toprule
& \multicolumn{4}{c}{\textbf{Low Confounding}} & \multicolumn{4}{c}{\textbf{High Confounding}} \\
\cmidrule(lr){2-5}\cmidrule(lr){6-9}
Estimator & MSE & Bias & Var & PS & MSE & Bias & Var & PS \\
\midrule
Masking ($\lambda_p$=10)   & .0571 & -.1945 & .0192 & 27.1\% & .0431 & .1183 & .0291 & 44.1\% \\
Masking ($\lambda_p$=30)   & .0428 & -.1468 & .0213 & 48.0\% & .0379 & -.0645 & .0337 & 65.6\% \\
Masking ($\lambda_p$=100)  & .0190 & -.0467 & .0168 & 93.5\% & .0440 & -.0887 & .0362 & 95.1\% \\
Masking ($\lambda_p$=300)  & .0148 & .0123 & .0146 & 100.0\% & .0398 & -.0219 & .0393 & 100.0\% \\
Masking ($\lambda_p$=1000) & .0202 & .0789 & .0140 & 100.0\% & .1050 & .2454 & .0448 & 100.0\% \\
\bottomrule
\end{tabular}
\caption{Continuous outcomes: fixed $\lambda_p$ sweep.}
\label{tab:fixed-cont}
\end{table}

\pagebreak

\clearpage 
\end{document}